\documentclass[10pt]{article} 
\usepackage{fullpage,amsmath,amsthm,amssymb,hyperref,cleveref,graphicx,url,rotating}
\usepackage{supertabular,algorithm2e} 

\graphicspath{{Fig/}}

\title{Hilbert geometry of the Siegel disk: The Siegel-Klein disk model}

\author{Frank Nielsen\\ Sony Computer Science Laboratories Inc, Tokyo, Japan}
\date{}

\begin{document}
\maketitle

\begin{abstract}
We study the Hilbert geometry induced by the Siegel disk domain, an open bounded convex set of complex square matrices of operator norm strictly less than one.
This Hilbert geometry yields a generalization of the Klein disk model of hyperbolic geometry, henceforth called the Siegel-Klein disk model to differentiate it with the classical Siegel upper plane and disk domains.
In the Siegel-Klein disk, geodesics are by construction always unique and Euclidean straight, allowing one to design
efficient geometric algorithms and data-structures from computational geometry.
For example, we show how to approximate the smallest enclosing ball of a set of complex square matrices in the Siegel disk domains:
We compare two generalizations of the iterative core-set algorithm of Badoiu and Clarkson (BC) in the Siegel-Poincar\'e disk and in the Siegel-Klein disk:
We demonstrate that geometric computing in the Siegel-Klein disk allows one (i) to bypass the time-costly recentering operations to the  disk origin required at each iteration of the BC algorithm in the Siegel-Poincar\'e disk model, and (ii) to approximate fast and numerically the Siegel-Klein distance with guaranteed lower and upper bounds derived from nested Hilbert geometries.
\end{abstract}


\def\dbarw{\mathrm{d}\bar{w}}
\def\dbarW{\mathrm{d}\bar{W}}
\def\FD{\mathbb{FD}}
\def\vec{\mathrm{vec}}
\def\SK{\mathbb{SK}}
\def\dP{\mathrm{d}P}
\def\SpO{\mathrm{SpO}}
\def\barK{\bar{K}}
\def\dW{\mathrm{d}W}
\def\PD{\mathrm{PD}}
\def\HPD{\mathrm{HPD}}
\def\ray{\mathrm{ray}}
\def\Exp{\mathrm{Exp}}
\def\CR{\mathrm{CR}}
\def\innerlabel#1#2#3{{\langle #1,#2\rangle_{#3}}}
\def\calX{\mathcal{X}}
\def\SD{\mathbb{SD}}
\def\SH{\mathbb{SH}}
\def\bbF{\mathbb{F}}
\def\arccosh{\mathrm{arccosh}}
\def\arctanh{\mathrm{arctanh}}
\def\PU{\mathbb{PU}}
\def\dY{\mathrm{d}Y}
\def\Im{\mathrm{Im}}
\def\Re{\mathrm{Re}}
\def\dx{\mathrm{d}x}
\def\dy{\mathrm{d}y}
\def\dw{\mathrm{d}w}
\def\dz{\mathrm{d}z}
\def\dwbar{\mathrm{d}\overline{w}}
\def\dzbar{\mathrm{d}\overline{z}}
\def\dbarz{\mathrm{d}\overline{z}}
\def\barz{\overline{z}}
\def\barw{\overline{w}}
\def\bbRP{\mathbb{RP}}
\def\PSL{\mathrm{PSL}}
\def\bbD{\mathbb{D}}
\def\ds{\mathrm{d}s}
\def\dS{\mathrm{d}S}
\def\dZ{\mathrm{d}Z}
\def\bbU{\mathbb{U}}
\def\dZbar{\mathrm{d}\bar{Z}}
\def\st{\ : \ }
\def\PSp{\mathrm{PSp}}
\def\Aut{\mathrm{Aut}}
\def\max{\mathrm{max}}
\def\min{\mathrm{min}}
\def\bbSU{\mathbb{SU}}
\def\bbS{\mathbb{S}}
\def\Moeb{\mathrm{Moeb}}
\def\Isom{\mathrm{Isom}}
\def\SU{\mathrm{SU}}
\def\SO{\mathrm{SO}}
\def\SL{\mathrm{SL}}
\def\Sp{\mathrm{Sp}}
\def\GL{\mathrm{GL}}
\def\Sym{\mathrm{Sym}}
\def\tr{\mathrm{tr}}
\def\bbR{\mathbb{R}}
\def\bbC{\mathbb{C}}
\def\bbF{\mathbb{F}}
\def\bb\SU{\mathbb{SU}}
\def\diag{\mathrm{diag}}
\def\mattwotwo#1#2#3#4{\left[\begin{array}{cc}#1 & #2\cr #3 & #4\end{array}\right]}
\def\tanh{\mathrm{tanh}}
\def\sinh{\mathrm{sinh}}
\def\cosh{\mathrm{cosh}}
\def\smap#1{{\langle #1\rangle}}
\def\bbP{\mathbb{P}}
\def\KL{\mathrm{KL}}
\def\barM{\overline{M}}
\def\barZ{\overline{Z}}
\def\barS{\overline{S}}
\def\barW{\overline{W}}
\def\innerlabel#1#2#3{{\langle #1,#2\rangle_{#3}}}
\def\Log{\mathrm{Log}}
\def\dP{\mathrm{d}P}

\def\bbH{\mathbb{H}}
\def\bbP{\mathbb{P}}
\def\bbSP{\mathbb{SP}}
\def\dt{\mathrm{d}t}
\def\Sym{\mathrm{Sym}}
\def\bbN{\mathbb{N}}
\def\tSigma{\tilde{\Sigma}}
\def\tmu{\tilde{\mu}}
\def\Log{\mathrm{Log}}

\def\st{\ :\ }
\def\Im{\mathrm{Im}}
\def\Re{\mathrm{Re}}
\def\C{\mathbb{C}}
\def\f12{{\frac{1}{2}}}
\def\bP{{\bar P}}
\def\bQ{{\bar Q}}
\def\bbC{\mathbb{C}}
\def\UP{\mathbb{UP}}
\def\Diag{\mathrm{Diag}}

\def\bbR{\mathbb{R}}
\def\bbC{\mathbb{C}}

\def\bbSH{\mathbb{SH}}
\def\bbSD{\mathbb{SD}}

\def\bbP{\mathbb{P}}
\def\bbX{\mathbb{X}}

\def\calN{\mathcal{N}}
\def\calM{\mathcal{M}}
\def\calP{\mathcal{P}}
\def\etal{et al.}

\def\defeq{:=}
\def\eqdef{=:}
\def\tr{\mathrm{tr}}
\def\Sym{\mathrm{Sym}}
\def\ball{\mathrm{ball}}
\def\SPD{\mathrm{SPD}}
\def\FR{\mathrm{FR}}
\def\CO{\mathrm{CO}}

\def\Spectrum{\mathrm{Spectrum}}
\def\Mult{\mathrm{Mult}}
\def\dv{\mathrm{d}v}

\def\eps{\epsilon}
\def\ceil#1{{\lceil #1\rceil}}
\def\dZ{\mathrm{d}Z}
\def\dZbar{\mathrm{d}\bar Z}
\def\ds{\mathrm{d}s}
\def\dx{\mathrm{d}x}

\def\Sp{\mathrm{Sp}}
\def\smap#1{{\langle {#1}\rangle}}

\def\dS{\dot{S}}
\def\bS{\bar{S}}
\def\on{\mathrm{on}}

\def\rank{\mathrm{rank}}
 \newtheorem{Problem}{Problem}
\newtheorem{Corollary}{Corollary}
\newtheorem{Lemma}{Lemma}
\newtheorem{Theorem}{Theorem}

\newtheorem{Definition}{Definition}
\newtheorem{Proposition}{Proposition}
\newtheorem{Property}{Property}
\newenvironment{`roof}{\paragraph{Proof:}}{\hfill$\square$}
 

\sloppy

\section{Introduction}

German mathematician Carl Ludwig Siegel~\cite{Siegel-1943} (1896-1981) and Chinese mathematician Loo-Keng Hua~\cite{Hua-1944} (1910-1985)    have introduced independently the {\em symplectic geometry} in the 1940's (with a preliminary work of Siegel~\cite{siegel1939einfuhrung} released in German in 1939).
The adjective {\em symplectic} stems from the greek, and means ``complex'': That is, mathematically the number field $\bbC$ instead of the ordinary real field $\bbR$.  
Symplectic geometry was originally motivated by the study of {\em complex multivariate functions} in the two landmark papers of Siegel~\cite{Siegel-1943} and Hua~\cite{Hua-1944}.
As we shall see soon, the naming ``symplectic geometry'' for the geometry of complex matrices originally stems from the relationships with the symplectic groups (and their matrix representations). Nowadays, 
symplectic geometry is mainly understood as the study of {\em symplectic manifolds}~\cite{blair2010riemannian} which are even-dimensional differentiable manifolds equipped with a closed and nondegenerate differential $2$-form $\omega$, called the {\em symplectic form}, studied in geometric mechanics.

We refer the reader to the PhD thesis~\cite{Freitas-1999,koufany2006analyse} for an overview of Siegel bounded domains.
More generally, the Siegel-like bounded domains have been studied and classified into $6$ types in the most general setting of {\em bounded symmetric irreducible homogeneous domains}  by Elie Cartan~\cite{Cartan-1935} in 1935 (see also~\cite{koszul1959exposes,berezin1975quantization}).

The Siegel upper space and the Siegel disk domains  provide generalizations of the complex Poincar\'e upper plane and the complex Poincar\'e disk to spaces of symmetric square complex matrices. 
In the remainder, we shall term them the {\em Siegel-Poincar\'e upper plane} and the {\em Siegel-Poincar\'e disk}.
The Siegel upper space includes the well-studied cone of real symmetric positive-definite (SPD) matrices~\cite{forstner2003metric} (SPD manifold). The celebrated affine-invariant SPD Riemannian metric~\cite{harandi2014manifold} can be recovered as a {\em special case} of the Siegel metric.

Applications of the geometry of Siegel upper/disk domains are found in radar processing~\cite{GeometrySPD-2008,Barbaresco-2011,Barbaresco-MIG-2013,barbaresco2013information} specially for dealing with Toepliz matrices~\cite{Jeuris-2016,ToepliztSiegel-2019}, probability density estimations~\cite{KDESiegel-2016} and probability metric distances~\cite{burbea1984informative,DistMVN-1990,CalvoOller-2002,tang2015information}, information fusion~\cite{tang2018information}, neural networks~\cite{RiemannTheta-2020}, theoretical physics~\cite{ohsawa2015siegel,froese2006transfer,ohsawa2017geometry}, and image morphology operators~\cite{SiegelDescriptor-2016}, just to cite a few.

In this paper, we extend the {\em Klein disk model}~\cite{richter2011perspectives} of the hyperbolic geometry to the Siegel disk domain by considering the {\em Hilbert geometry}~\cite{Hilbert-1895}
 induced by the open bounded convex Siegel disk~\cite{papadopoulos2014handbook,HilbertSPD-1994}.
We call the Hilbert metric distance of the Siegel disk the {\em Siegel-Klein distance}.
We term this model the {\em Klein-Siegel model} for short to contrast it with the Poincar\'e-Siegel upper plane model and 
the Poincar\'e-Siegel disk model.
The main advantages of using the Klein-Siegel disk model instead of the usual Siegel-Poincar\'e upper plane or the Siegel-Poincar\'e disk  
 are that the geodesics are unique and always {\em straight} by construction.
Thus this Siegel-Klein disk model
is very well-suited for designing efficient algorithms and data-structures by borrowing techniques of Euclidean computational geometry~\cite{boissonnat1998algorithmic}.
Moreover, in the Siegel-Klein disk model, we have an {\em efficient and robust method} to {\em approximate with guarantees} the calculation of the Siegel-Klein distance: This is specially useful when handling high-dimensional square complex matrices.
The algorithmic advantage of the Hilbert geometry was already observed for real hyperbolic geometry (included as a special case of the Siegel-Klein model):
For example, the hyperbolic Voronoi diagrams can be efficiently computed as an affine power diagram clipped to the boundary circle~\cite{HVD-2010,CKclassification-2016,nielsen2014visualizing,nielsen2012hyperbolic}.
To demonstrate the advantage of the Siegel-Klein disk model (Hilbert distance) over the Siegel-Poincar\'e disk model (Kobayashi distance), we consider approximating the Smallest Encloding Ball (SEB) of the a set of square complex matrices in the Siegel disk domain.
This problem finds potential applications in {\em image morphology}~\cite{angulo2014morphological,SiegelDescriptor-2016} or {\em anomaly detection} of covariance matrices~\cite{tavallaee2008novel,cont2010information}.
Let us state the problem as follows:

\begin{Problem}[Smallest Enclosing Ball (SEB)]
Given a metric space $(X,\rho)$ and a finite set $\{p_1,\ldots, p_n\}$ of $n$ points in $X$, find the smallest-radius enclosing ball with 
 circumcenter $c^*$ minimizing the following objective function:
\begin{equation}
\min_{c\in X}\max_{i\in \{1,\ldots,n\}}\ \rho(c,p_i).
\end{equation}
\end{Problem}

In general, the SEBs may not be unique in a metric space: 
For example, the SEBs are not unique in a  {\em discrete Hamming metric space}~\cite{mazumdar2013chebyshev} making it notably NP-hard to calculate.
We note in passing that the set-complement of a Hamming ball is a Hamming ball in a Hamming metric space.
However, the SEB is proven unique in the Euclidean geometry~\cite{welzl1991smallest}, the hyperbolic geometry~\cite{hyperbolicSEB-2015}, the Riemannian positive-definite matrix manifold~\cite{lang2012math,MIG-2013}, and more generally in any {\em Cartan-Hadamard manifold}~\cite{arnaudon2013approximating} (Riemannian manifold that is complete and simply connected with non-positive sectional curvatures).
The SEB is guaranteed to be unique in any {\em Bruhat-Tits space}~\cite{lang2012math} (i.e., complete metric space with a semi-parallelogram law) which includes the Riemannian SPD manifold.

A fast $(1+\epsilon)$-approximation algorithm which requires  $\left\lceil{\frac{1}{\epsilon^2}}\right\rceil$ iterations was reported in~\cite{badoiu2003smaller,arnaudon2013approximating} to approximate the SEB in the Euclidean space: That is a {\em covering ball} of radius $(1+\epsilon)r^*$ where $r^*=\max_{i\in \{1,\ldots,n\}}\ \rho(c^*,p_i)$ for $c^*=\arg\min_{c\in X}\max_{i\in \{1,\ldots,n\}}\ \rho(c,p_i)$.
Since the approximation factor does {\em not} depend on the dimension, this SEB approximation algorithm found many applications in machine learning~\cite{tsang2006generalized} (e.g., in 
Reproducing Kernel Hilbert Spaces~\cite{RKHS-2016}, RKHS).

\subsection{Paper outline and contributions}

In Section~\ref{sec:hyperbolicgeo}, we concisely recall the usual models of the {\em hyperbolic complex plane}: 
The {\em Poincar\'e upper plane} model,
and the {\em Poincar\'e disk model}, and the {\em Klein disk model}.
We then briefly review the geometry of the Siegel upper plane domain in~\S\ref{sec:SiegelUpper} and the Siegel disk domain in~\S\ref{sec:SiegelDisk}.
Section~\ref{sec:SiegelKlein} introduces the novel {\em Siegel-Klein model}  using the Hilbert geometry and its Siegel-Klein distance.
To demonstrate the algorithmic advantage of using the Siegel-Klein disk model over the Siegel-Poincar\'e disk model in practice, 
we compare in~\ref{sec:SEB} the two implementations of the Badoiu and Clarkson's SEB approximation algorithm~\cite{badoiu2003smaller} in these models.
Finally, we conclude this work in \S\ref{sec:concl}.
In the Appendix, we first list the notations used in this work, recall the deflation method for calculating numerically the eigenvalues of a  Hermitian matrix (\S\ref{sec:deflation}), and provide some basic snippet code for calculating the Siegel distance (\S\ref{sec:snippet}).

Our main contributions are summarized as follows:

\begin{itemize}

\item  First, we formulate a generalization of the Klein disk model of hyperbolic geometry to the Siegel disk domain in Definition~\ref{def:SiegelKlein} using the framework of Hibert geometry. 
 We report the formula of the Siegel-Klein distance to the origin in~Theorem~\ref{thm:SKdistorigin} (and more generally a closed-form expression for the Siegel-Klein distance between two points whose supporting line passes through the origin), describe how to convert the Siegel-Poincar\'e disk to the Siegel-Klein disk and vice versa in Proposition~\ref{prop:convPK}, report an exact algorithm to calculate the Siegel-Klein distance for diagonal matrices in Theorem~\ref{thm:SKdiagformula}.
In practice, we show how to obtain a {\em fast guaranteed approximation of the Siegel-Klein distance} using geodesic bisection searches with guaranteed lower and upper bounds (Theorem~\ref{prop:LBboundSK} whose proof is obtained by considering nested Hilbert geometries).

\item Second, we report the exact solution to a geodesic cut problem in the Siegel-Poincar\'e/Siegel-Klein disks in Proposition~\ref{prop:SPgeodesicOrigin}.
This result yields an explicit equation for the geodesic linking the origin of the Siegel disk domain to any other matrix point of the Siegel disk domain (Proposition~\ref{prop:SPgeodesicOrigin} and Proposition~\ref{prop:SKgeodesicOrigin}).
We then report an implementation of the Badoiu and Clarkson's iterative algorithm~\cite{badoiu2003smaller} for approximating the smallest enclosing ball tailored to the Siegel-Poincar\'e and Siegel-Klein disk domains. 
In particular, we show in~\S\ref{sec:SEB} that the implementation in the Siegel-Klein model yields a fast algorithm which bypasses the  costly operations of recentering to the origin required in the Siegel-Poincar\'e disk model.
\end{itemize}

Let us now introduce a few notations on matrices and their norms.

\subsection{Matrix spaces and matrix norms}

Let $\bbF$ be a {\em number field} considered in the remainder to be either  the {\em real number field} $\bbR$ or the {\em complex number field} $\bbC$.
For a complex number $z=a+ib\in\bbC$ (with imaginary number $i^2=-1$), we denote by $\barz=a-ib$ its {\em complex conjugate}, and  by $|z|=\sqrt{z\barz}=\sqrt{a^2+b^2}$ its {\em modulus}.
Let $\Re(z)=a$ and $\Im(z)=b$ denote the {\em real part} and the {\em imaginary part} of the complex number $z=a+ib$, respectively.

Let $M(d,\bbF)$ be the space of $d\times d$ square matrices with coefficients in $\bbF$,  and let $\GL(d,\bbF)$ denote its subspace of invertible matrices.  Let $\Sym(d,\bbF)$ denote the vector space of $d\times d$ symmetric matrices with coefficients in $\bbF$.
The identity matrix is denoted by $I$ (or $I_d$ when we want to emphasize its $d\times d$ dimension). 
The conjugate of a matrix $M=[M_{i,j}]_{i,j}$ is the matrix of complex conjugates: $\barM:=[\barM_{i,j}]_{i,j}$.
The {\em conjugate transpose} of a matrix $M$ is $M^H=(\bar{M})^\top=\overline{M^\top}$, the adjoint matrix.
Conjugate transposition is also denoted by the star operator (i.e., $M^*$) or the dagger symbol (i.e., $M^\dagger$) in the literature.
A complex matrix is said {\em Hermitian} when $M^H=M$ (hence $M$ has real diagonal elements). For any $M\in M(d,\bbC)$, 
Matrix $MM^H$ is Hermitian: $(MM^H)^H=(M^H)^H(M)^H=MM^H$.

A real matrix $M\in M(d,\bbR)$ is said {\em symmetric positive-definite} (SPD) if and only if $x^\top Mx>0$ for all $x\in\bbR^d$ with $x\not=0$.
This positive-definiteness property is written $M\succ 0$, where $\succ$ denotes the partial {\em L\"owner ordering}~\cite{nielsen2017fast}.
Let $\PD(d,\bbR)=\{ P\succ 0 \ :\ P\in\Sym(d,\bbR)\}$ be the space of real symmetric positive-definite matrices~\cite{forstner2003metric,Moakher-2005,lang2012math,niculescu2018convex} of dimension $d\times d$. 
This space is not a vector space but a {\em cone}, i.e., if $P_1,P_2\in\PD(d,\bbR)$ then $P_1+\lambda P_2\in\PD(d,\bbR)$ for all $\lambda>0$.
The boundary of the cone consists of rank-deficient symmetric positive semi-definite matrices.

The (complex/real) {\em eigenvalues} of a square complex matrix $M$ are ordered such that $|\lambda_1(M)| \geq \ldots \geq |\lambda_d(M)|$, where $|\cdot|$ denotes the complex modulus. The {\em spectrum} $\lambda(M)$ of a matrix $M$ is its set of eigenvalues: 
$\lambda(M)=\{\lambda_1(M),\ldots,\lambda_d(M) \}$.
In general, real matrices may have complex eigenvalues but symmetric matrices (including SPD matrices) have always real eigenvalues.
The {\em singular values} $\sigma_i(M)$ of $M$ are always real: 
\begin{equation}
\sigma_i(M)=\sqrt{\lambda_i(M\barM)}=\sqrt{\lambda_i(\barM M)},
\end{equation}
 and ordered as follows: $\sigma_1(M)\geq \ldots\geq \sigma_d(M)$ with $\sigma_\max(M)=\sigma_1(M)$ and $\sigma_\min(M)=\sigma_d(M)$.
We have $\sigma_{d-i+1}(M^{-1})=\frac{1}{\sigma_i(M)}$, and in particular $\sigma_d(M^{-1})=\frac{1}{\sigma_1(M)}$.

Any {\em matrix norm} $\|\cdot\|$  (including the operator norm) satisfies:
\begin{itemize}
\item $\|M\|\geq 0$ with equality if and only if $M=0$ (where $0$ denotes the matrix with all its entries equal to zero),
\item $\|\alpha M\|=|\alpha| \|M\|$,  
\item $\|M_1+M_2\|\leq \|M_1\|+\|M_2\|$, and
\item $\|M_1 M_2\|\leq \|M_1\|\ \|M_2\|$.
\end{itemize}
Let us define two usual matrix norms: The Fr\"obenius norm and the operator norm.
The {\em Fr\"obenius norm} of $M$ is: 
\begin{eqnarray}
\|M\|_F &:=& \sqrt{\sum_{i,j} |M_{i,j}|^2},\\
&=& \sqrt{\tr(M M^H)} = \sqrt{\tr(M^H M)}.
\end{eqnarray}
The induced Fr\"obenius distance between two complex matrices $C_1$ and $C_2$ is
$\rho_E(C_1,C_2)=\|C_1-C_2\|_F$.

The {\em operator norm} or {\em spectral norm} of a matrix $M$ is:
\begin{eqnarray}
\|M\|_O &=& \max_{x\not=0} \frac{\|Mx\|_2}{\|x\|_2},\\
&=& \sqrt{\lambda_{\max}(M^H M)},\\
&=& \sigma_\max(M).
\end{eqnarray}
Notice that $M^H M$ is a Hermitian positive semi-definite matrix.
The operator norm coincides with the {\em spectral radius} $\rho(M)=\max_i \{|\lambda_i(M)|\}$ of the matrix $M$ and 
is upper bounded by the Fr\"obenius norm: $\|M\|_O\leq \|M\|_F$, and we have $\|M\|_O\geq \max_{i,j} |M_{i,j}|$.
When the dimension $d=1$, the operator norm of $[M]$ coincides with the complex modulus: $\|M\|_O=|M|$.

To calculate the largest singular value $\sigma_\max$, we may use a the (normalized) {\em power method}~\cite{lanczos1950iteration,cullum2002lanczos} which has quadratic convergence for Hermitian matrices (see Appendix~\ref{sec:deflation}). We can also use the more costly {\em Singular Value Decomposition} (SVD) of $M$ which requires cubic time: $M=UDV^H$ where $D=\Diag(\sigma_1,\ldots,\sigma_d)$ is the diagonal matrix with coefficients being the singular values of $M$.

\section{Hyperbolic geometry in the complex plane: The Poincar\'e upper plane and disk models and the Klein disk model}\label{sec:hyperbolicgeo}

We concisely review the three usual models of the hyperbolic plane~\cite{cannon1997hyperbolic,goldman1999complex}:
 Poincar\'e upper plane model in \S\ref{sec:PoincareUpper}, the Poincar\'e disk model in~\S\ref{sec:PoincareDisk}, and the Klein disk model in~\S\ref{sec:KleinDisk}. We then report distance expressions in these models and conversions between these three usual models in~\S\ref{sec:conversions}.
Finally in~\S\ref{sec:ighyp}, we recall the important role of hyperbolic geometry in the Fisher-Rao geometry in 
information geometry~\cite{IG-2016,nielsen2018elementary}.

\subsection{Poincar\'e complex upper plane}\label{sec:PoincareUpper}

The Poincar\'e upper plane domain is defined by
\begin{equation}
\bbH = \left\{ z=a+ib \st z\in\bbC, b=\Im(z)>0 \right\}.
\end{equation}

The Hermitian metric tensor is:
\begin{equation}
\ds^2_U=\frac{\dz\dzbar}{\Im(z)^2},
\end{equation}
or equivalently the Riemannian line element is:
\begin{equation}
\ds^2_U=\frac{\dx^2+\dy^2}{y^2},
\end{equation}

Geodesics between $z_1$ and $z_2$ are either arcs of semi-circles whose centers are located on the  real axis and orthogonal to the real axis, or vertical line segments when $\Re(z_1)=\Re(z_2)$.

The geodesic length distance is
\begin{equation}\label{eq:PoincareDistLogGeneric}
\rho_{U}(z_1,z_2) := \log \left(\frac{|z_1-\barz_2|+|z_1-z_2| }{|z_1-\barz_2|-|z_1-z_2|}\right),
\end{equation}
or equivalently
\begin{equation}
\rho_U(z_1,z_2) = \arccosh\left( \sqrt{ \frac{|z_1-\barz_2|^2}{\Im(z_1)\Im(z_2)}} \right),
\end{equation}
where
\begin{equation}
\mathrm{arccosh}(x)=\log \left(x+\sqrt{x^{2}-1}\right), \quad x\geq 1.
\end{equation}

Equivalent formula can be obtained by using the following identity:
\begin{equation}
\log(x)=\operatorname{arcosh}\left(\frac{x^{2}+1}{2 x}\right)=\operatorname{artanh}\left(\frac{x^{2}-1}{x^{2}+1}\right),
\end{equation}
where
\begin{equation}
\operatorname{artanh}(x)=\frac{1}{2}\log \left(\frac{1+x}{1-x}\right),\quad x<1.
\end{equation}

By interpreting a complex number $z=x+iy$ as a 2D point with Cartesian coordinates $(x,y)$, the metric can be rewritten as
\begin{equation}\label{eq:upconformal}
\ds^2_U=\frac{\dx^2+\dy^2}{y^2}= \frac{1}{y^2}\ds^2_E,
\end{equation}
where $\ds^2_E=\dx^2+\dy^2$ is the Euclidean (flat) metric.
That is, the Poincar\'e upper plane metric $\ds_U$ can be rewritten as a {\em conformal factor} $\frac{1}{y}$ times the Euclidean metric $\ds_E$.
Thus the metric of Eq.~\ref{eq:upconformal} shows that the Poincar\'e upper plane model is a {\em conformal model} of hyperbolic geometry: 
That is, the Euclidean angle measurements in the $(x,y)$ chart coincide with the underlying hyperbolic angles. 

The group of orientation-preserving isometries (i.e., without reflections)  is the {\em real projective special group}
 $\PSL(2,\bbR)=\SL(2,\bbR)/\{\pm I\}$ (quotient group), where $\SL(2,\bbR)$ denotes the {\em special linear group} of matrices with unit determinant:
\begin{equation}
\Isom^+(\bbH)\cong\PSL(2,\bbR).
\end{equation}

The left group action is a fractional linear transformation (also called a M\"obius transformation): 
\begin{equation}
g.z=\frac{a z+b}{c z+d},\quad g=\mattwotwo{a}{b}{c}{d},\quad ad-bc\not=0.
\end{equation}
The condition $ab-cd\not=0$ is to ensure that the  M\"obius transformation is not constant.
The set of M\"obius transformations form a group $\Moeb(\bbR,2)$.
The elements of the M\"obius group can be represented by corresponding $2\times 2$ matrices of $\PSL(2,\bbR)$: 
\begin{equation}
\left\{\mattwotwo{a}{b}{c}{d},\quad ad-bc\not =0\right\}.
\end{equation}
The neutral element $e$ is encoded by the identity matrix $I$.

The fractional linear transformations 
\begin{equation}
w(z) = \frac{a z+b}{c z+d},\quad   a,b,c,d\in\bbR,ad-bc\not=0
\end{equation}
are the analytic mappings $\bbC\cup\{\infty\}\rightarrow \bbC\cup\{\infty\}$ of the Poincar\'e upper plane onto itself.

The group action is {\em transitive} (i.e., $\forall z_1,z_2\in\bbH, \exists g$ such that $g.z_1=z_2$) and {\em faithful} (i.e., if $g.z=z\forall z$ then $g=e$). The {\em stabilizer} of $i$  is the rotation group:
\begin{equation}
\mathrm{SO}(2)=\left\{\left[\begin{array}{cc}
\cos \theta & \sin \theta \\
-\sin \theta & \cos \theta
\end{array}\right] \ :\  \theta \in \mathbb{R}\right\}.
\end{equation}

The unit speed geodesic anchored at $i$ and going upward (i.e., geodesic with initial condition) is:
\begin{equation}
\gamma(t)=\left[\begin{array}{cc}
e^{t / 2} & 0 \\
0 & e^{-t / 2}
\end{array}\right] \times i= i e^{t}.
\end{equation}

Since the other geodesics can be obtained by the action of $\PSL(2,\bbR)$, it follows that the geodesics in $\bbH$ are parameterized by:
\begin{equation}
\gamma(t)=\frac{a i e^{t}+b}{c i e^{t}+d}.
\end{equation}

\subsection{Poincar\'e disk}\label{sec:PoincareDisk}

The Poincar\'e unit disk is
\begin{equation}
\bbD= \left\{ \barw w<1 \st w\in\bbC \right\}.
\end{equation}

The Riemannian Poincar\'e line element (also called Poincar\'e-Bergman line element) is
\begin{equation}\label{eq:PoincareDiskMetric}
\ds^2_D=\frac{4\dw\dwbar}{(1-|w|^2)^2}.
\end{equation}

Since $\ds^2_D=\left(\frac{2}{1-\|x\|^2}\right)^2\ds_E^2$, we deduce that the metric is conformal: The Poincar\'e disk is a conformal model of hyperbolic geometry.
The geodesic between points $w_1$ and $w_2$ are  either arcs of circles intersecting orthogonally the disk boundary $\partial\bbD$, or straight lines passing through the origin $0$ of the disk and clipped to the disk domain.

The geodesic distance in the Poincar\'e disk is
\begin{eqnarray}
\rho_D(w_1,w_2) &=& \arccosh\left( \sqrt{\frac{|w_1\barw_2-1|^2}{(1-|w_1|^2)(1-|w_2|^2)}}\right),\\
 &=& 2\ \arctanh\left| \frac{w_2-w_1}{1-\barw_1w_2}  \right|.
\end{eqnarray}

The group of orientation preserving isometry is the {\em complex projective special group} $\PSL(2,\bbC)=\SL(2,\bbC)/\{\pm I\}$ where 
$\SL(2,\bbC)$ denotes the special group of $2\times 2$ complex matrices with unit determinant.

In the Poincar\'e disk model, the transformation
\begin{equation}\label{eq:transrot}
T_{z_0,\theta}(z) = e^{i\theta} \frac{z-z_0}{1-\barz_0z}
\end{equation}
 corresponds to a {\em hyperbolic motion} (a M\"obius transformation~\cite{ratcliffe1994foundations}) which moves point $z_0$ to the origin $0$, and then makes a rotation of angle $\theta$.
The group of such transformations is the {\em automorphism group} of the disk, $\Aut(\bbD)$,
 and the transformation $T_{z_0,\theta}$ is called a biholomorphic automorphism (i.e., a one-to-one conformal mapping of the disk onto itself).

The Poincar\'e distance is invariant under automorphisms of the disk, and more generally the Poincar\'e distance decreases under holomorphic mappings (Schwarz–Pick theorem): That is, the Poincar\'e distance is contractible under holomorphic mappings $f$:
 $\rho_D(f(w_1),f(w_2))\leq \rho_D(w_1,w_2)$.

\subsubsection{Klein disk}\label{sec:KleinDisk}

The Klein disk model~\cite{cannon1997hyperbolic,richter2011perspectives} (also called the Klein-Beltrami model) is defined on the unit disk domain as the Poincar\'e disk model.
The Klein metric is
\begin{equation}
\ds_K^2 = \left(\frac{\mathrm{d} s_{E}^{2}}{1-\|x\|_{E}^{2}}+\frac{\langle x, \mathrm{d} x\rangle_{\mathrm{E}}}{\left(1-\|x\|_{E}^{2}\right)^{2}}\right).
\end{equation}
It is {\em not} a conformal metric (except at the disk origin), and therefore the Euclidean angles in the $(x,y)$ chart do not correspond to the  underlying hyperbolic angles.

The Klein distance between two points $k_1=(x_1,y_1)$ and $k_2=(x_2,y_2)$ is
\begin{equation}\label{eq:DistKlein1D}
\rho_K(k_1,k_2) = \mathrm{arccosh}\left(
\frac{1-(x_1x_2+y_1y_2)}{\sqrt{(1-\|k_1\|^2)(1-\|k_2\|^2)}}
\right).
\end{equation}
An equivalent formula shall be reported later in page~\pageref{proofKlein1d} in a more setting of Theorem~\ref{thm:SKdiagformula}.

The advantage of the Klein disk over the Poincar\'e disk is that geodesics are {\em straight Euclidean lines clipped to the unit disk domain}.
Therefore this model is well-suited to implement computational geometric algorithms and data structures, see for example~\cite{HVD-2010,jin2018conformal}.
The group of isometries in the Klein model are projective maps $\bbRP^2$ preserving the disk. 
We shall see that the Klein disk model corresponds to the Hilbert geometry of the unit disk.

\subsection{Poincar\'e and Klein distances to the disk origin and conversions}\label{sec:conversions}

In the Poincar\'e disk, the distance of a point $w$ to the origin $0$ is
\begin{equation}\label{eq:PoincareDiskOrigin}
\rho_D(0,w) = \log \left(\frac{1+|w|}{1-|w|}\right).
\end{equation}

Since the Poincar\'e disk model is conformal (and M\"obius transformations are conformal maps), Eq.~\ref{eq:PoincareDiskOrigin} shows that 
Poincar\'e disks have Euclidean disk shapes (however with displaced centers).

In the Klein disk, the distance of a point $k$ to the origin is
\begin{equation}\label{eq:kleindistorigin}
\rho_K(0,k) = \frac{1}{2}\log \left(\frac{1+|k|}{1-|k|}\right)=\frac{1}{2}\rho_D(0,k).
\end{equation}
Observe the multiplicative factor of $\frac{1}{2}$ in Eq.~\ref{eq:kleindistorigin}.

Thus we can easily convert a point $w\in\bbC$ in the Poincar\'e disk to a point $k\in\bbC$ in the Klein disk, and vice-versa as follows:

\begin{eqnarray}
w &=&\frac{1}{1+\sqrt{1-|k|^2}}\ k,\\
k &=&\frac{2}{1+|w|^2}\ w.
\end{eqnarray}

Let $C_{K\rightarrow D}(k)$ and $C_{D\rightarrow K}(w)$ denote these conversion functions with
\begin{eqnarray}
C_{K\rightarrow D}(k) &=& \frac{1}{1+\sqrt{1-|k|^2}}\ k,\\
C_{D\rightarrow K}(w) &=& \frac{2}{1+|w|^2}\ w.
\end{eqnarray}

We can write $C_{K\rightarrow D}(k)=\alpha(k) k$ and $C_{D\rightarrow K}(w)=\beta(w) w$, 
so that $\alpha(k)>1$ is an {\em expansion factor}, and $\beta(w)<1$ is a {\em contraction factor}.

The conversion functions are M\"obius transformations represented by the following matrices:
\begin{eqnarray}
M_{K\rightarrow D}(k) &=& \mattwotwo{\alpha(k)}{0}{0}{1},\\
M_{D\rightarrow K}(w) &=& \mattwotwo{\beta(w)}{0}{0}{1}.
\end{eqnarray}

For sanity check, let $w=r+0i$ be a point in the Poincar\'e disk with equivalent point $k=\frac{2}{1+r^2}r+0i$ in the Klein disk.
Then we have:
\begin{eqnarray}
\rho_K(0,k) &=& \frac{1}{2}\log \left(\frac{1+|k|}{1-|k|}\right),\\
&=&\frac{1}{2}\log \left(\frac{1+\frac{2}{1+r^2}r}{1-\frac{2}{1+r^2}r} \right),\\
&=& \frac{1}{2}\log \left(\frac{1+r^2+2r}{1+r^2-2r}\right),\\
&=& \frac{1}{2}\log \left(\frac{(1+r)^2}{(1-r)^2}\right),\\
&=& \log \left(\frac{1+r}{1-r}\right) = \rho_D(0,w).
\end{eqnarray}

We can convert a point $z$ in the Poincar\'e upper plane to a corresponding point $w$ in the Poincar\'e disk, or vice versa, using the following M\"obius transformations:

\begin{eqnarray}
w&=&\frac{z-i}{z+i},\\
z&=&i\frac{1+w}{1-w}.
\end{eqnarray}

Notice that we compose M\"obius transformations by multiplying their matrix representations.

\subsection{Hyperbolic Fisher-Rao geometry of location-scale families}\label{sec:ighyp}

Consider a parametric  family $\mathcal{P}=\{p_\theta(x)\}_{\theta\in\Theta}$  of probability densities dominated by a positive measure $\mu$ (usually, the Lebesgue measure or the counting measure) defined on a measurable space $(\calX,\Sigma)$, where $\calX$ denotes the support of the densities and $\Sigma$ is a finite $\sigma$-algebra~\cite{IG-2016}.
Hotelling~\cite{Hotelling-1930} and Rao~\cite{Rao-1945} independently considered the Riemannian geometry of $\mathcal{P}$ by using the Fisher Information Matrix (FIM) to define the Riemannian metric tensor~\cite{CRLB-2013} expressed in the (local) coordinates $\theta\in\Theta$, where $\Theta$ denotes the parameter space.
The FIM is defined by the following symmetric positive semi-definite matrix~\cite{IG-2016,RFIM-2017}:
\begin{equation}
I(\theta)=E_{p_\theta}\left[\nabla_\theta \log p_\theta(x) \left(\nabla_\theta \log p_\theta(x)\right)^\top\right].
\end{equation}
When $\cal{P}$ is regular~\cite{IG-2016}, the FIM is guaranteed to be positive-definite, and can thus play the role of a metric tensor field: The so-called {\em Fisher metric}.

Consider the location-scale family induced by a density $f(x)$ symmetric with respect to $0$ such that $\int_\calX f(x)\mathrm{d}\mu(x)=1$, 
$\int_\calX xf(x)\mathrm{d}\mu(x)=0$ and $\int_\calX x^2f(x)\mathrm{d}\mu(x)=1$ (with $\calX=\bbR$):
\begin{equation}
\mathcal{P}=\left\{p_\theta(x)=\frac{1}{\theta_2}f\left(\frac{x-\theta_1}{\theta_2}\right),\quad \theta=(\theta_1,\theta_2)\in \bbR\times\bbR_{++}\right\}.
\end{equation}
The density $f(x)$ is called the standard density, and corresponds to the parameter $(0,1)$: $p_{(0,1)}(x)=f(x)$.
The parameter space $\Theta=\bbR\times\bbR_{++}$ is the upper plane, and the FIM can be structurally calculated~\cite{komaki2007bayesian}
 as the following diagonal matrix:
\begin{equation}
I(\theta)=\left[\begin{array}{ll} a^2 & 0\cr 0 & b^2\end{array}\right]
\end{equation}
with
\begin{eqnarray}
a^{2}:=\int\left(\frac{f^{\prime}(x)}{f(x)}\right)^{2} f(x) \mathrm{d}\mu(x),\\
b^{2}:=\int\left(x \frac{f^{\prime}(x)}{f(x)}+1\right)^{2} f(x) \mathrm{d}\mu(x).
\end{eqnarray}

By rescaling $\theta=(\theta_1,\theta_2)$ as $\theta'=(\theta_1',\theta_2')$ with  $\theta_1'=\frac{a}{b\sqrt{2}}\theta_1$ and $\theta_2'=\theta_2$, we get the FIM with respect to $\theta'$ expressed as:
\begin{equation}
I(\theta)=\frac{b^2}{\theta_2^2} \left[\begin{array}{ll} 1 & 0\cr 0 & 1\end{array}\right],
\end{equation}
a constant time the Poincar\'e metric in the upper plane.
Thus the Fisher-Rao manifold of  a location-scale family (with symmetric standard density $f$) is isometric to the planar hyperbolic space of negative curvature $\kappa=-\frac{1}{b^2}$.

\section{The Siegel upper space and the Siegel distance}\label{sec:SiegelUpper}

The {\em Siegel upper space}~\cite{siegel1939einfuhrung,Siegel-1943,namikawa1980siegel,Freitas-1999} $\SH(d)$ is defined as the space of symmetric complex square matrices of size $d\times d$ which have positive-definite imaginary part:

\begin{equation}
\SH(d) := \left\{ Z = X + iY \st X\in\Sym(d,\bbR), Y\in\PD(d,\bbR)\right\}.
\end{equation}

The space $\SH(d)$ is a tube domain of dimension $d(d+1)$ since 
\begin{equation}
\dim(\SH(d))=\dim(\Sym(d,\bbR))+\dim(\PD(d,\bbR)),
\end{equation}
 with 
$\dim(\Sym(d,\bbR))=\frac{d(d+1)}{2}$ and $\dim(\PD(d,\bbR))=\frac{d(d+1)}{2}$.
We can extract the components $X$ and $Y$ from $Z$ as $X=\frac{1}{2}(Z+\bar Z)$ and $Y=\frac{1}{2i}(Z-\bar Z)=-\frac{i}{2}(Z-\bar Z)$.
The matrix pair $(X,Y)$ belongs to the Cartesian product of a matrix vector space with the symmetric positive-definite (SPD) matrix cone: $(X,Y)\in\Sym(d,\bbR)\times \PD(d,\bbR)$.
When $d=1$, the Siegel upper space coincides with the Poincar\'e upper plane: $\SH(1)=\bbH$.
The geometry of the Siegel upper space was studied independently by Siegel~\cite{Siegel-1943}  and  Hua~\cite{Hua-1944} from different viewpoints in the late 1930's--1940's.
Historically, these classes of complex matrices $Z\in\SH(d)$ were first studied by Riemann~\cite{Riemanntheorie-1857}, 
and later eponymously  called
 {\em Riemann matrices}. 
Riemann matrices are used to define Riemann theta functions~\cite{RiemannTheta-1865,ThetaSage-2016,ThetaJulia-2019,agostini2019discrete}.

The {\em Siegel distance} in the upper plane is induced by the following line element:
\begin{equation}
\ds_U^2 (Z) = 2 \tr\left(Y^{-1}\dZ\ Y^{-1}\dZbar\right).
\end{equation}

The formula for the Siegel upper distance between $Z_1$ and $Z_2\in\SH(d)$ was calculated in Siegel's masterpiece paper~\cite{Siegel-1943} as follows:
\begin{equation}\label{eq:SiegelDistance}
\rho_U(Z_1,Z_2) = \sqrt{\sum_{i=1}^d \log^2\left(\frac{1+\sqrt{r_i}}{1-\sqrt{r_i}}\right)},
\end{equation}
where
\begin{equation}
r_i=\lambda_i\left(R(Z_1,Z_2)\right),
\end{equation}
with $R(Z_1,Z_2)$ denoting  the matrix generalization~\cite{MatrixCrossRationRiccati-1993} of the {\em cross-ratio}:
\begin{equation}\label{eq:matrixcr}
R(Z_1,Z_2)  := (Z_1-Z_2)(Z_1-\barZ_2)^{-1} (\barZ_1 -\barZ_2) (\barZ_1 -Z_2)^{-1},
\end{equation}
and $\lambda_i(M)$ denotes the $i$-th largest (real) eigenvalue of (complex) matrix $M$.
The letter notation 'R' in  $R(Z_1,Z_2)$ is a mnemonic which stands for 'r'atio.

The Siegel distance can also be expressed without explicitly using the eigenvalues as:
\begin{equation}\label{eq:SiegelDistancePowSeries}
\rho_U(Z_1,Z_2) = 2\sqrt{\tr\left(R_{12} \left( \sum_{i=0}^\infty \frac{R_{12}^i}{2i+1} \right)^2\right) },
\end{equation}
where $R_{12}=R(Z_1,Z_2)$.
In particular, we can {\em truncate} the matrix power series of Eq.~\ref{eq:SiegelDistancePowSeries} to get an approximation of the Siegel distance:
\begin{equation}\label{eq:SiegelDistanceTruncPowSeries}
\tilde\rho_{U,l}(Z_1,Z_2) = 2\sqrt{\tr\left(R_{12} \left( \sum_{i=0}^l \frac{R_{12}^i}{2i+1} \right)^2\right) }.
\end{equation}

It costs $O(\Spectrum(d))=O(d^3)$ to calculate the Siegel distance using Eq.~\ref{eq:SiegelDistance} and
 $O(l\Mult(d))=O(l d^{2.3737})$ to approximate it using the truncated series formula of Eq.~\ref{eq:SiegelDistanceTruncPowSeries}, where $\Spectrum(d)$ denotes the cost of performing the spectral decomposition of a $d\times d$ complex matrix, and $\Mult(d)$ denotes the cost of multiplying two $d\times d$ square complex matrices. 
For example, choosing the Coppersmith-Winograd algorithm for $d\times d$ matrix multiplications, we have $\Mult(d)=O(d^{2.3737})$.
Although Siegel distance formula of Eq.~\ref{eq:SiegelDistanceTruncPowSeries} is attractive, the number of iterations $l$ to get an $\epsilon$-approximation of the Siegel distance depends on the dimension $d$. 
In practice, we can define a threshold $\delta>0$, 
and as a rule of thumb iterate on the truncated sum until $\left|\tr\left(\frac{R_{12}^i}{2i+1}\right)\right|<\delta$.

A {\em spectral function}~\cite{niculescu2018convex} of a matrix $M$ is a function $F$ which  is the composition of a
 symmetric function $f$ with the eigenvalue map $\Lambda$: $F(M):=(f\circ\Lambda)(M)=f(\Lambda(M))$.
For example, the Kullback-Leibler divergence between two zero-centered Gaussian distributions is a spectral function distance 
since we have:
\begin{eqnarray}
D_{\KL}(p_{\Sigma_1},p_{\Sigma_2})&=& \int p_{\Sigma_1}(x)\log\frac{p_{\Sigma_1}(x)}{p_{\Sigma_2}(x)} \mathrm{d}x,\\
&=& \frac{1}{2}\left( \log \frac{|\Sigma_2|}{|\Sigma_1|}+\tr(\Sigma_2^{-1}\Sigma_1)-d \right),\\
&=& \frac{1}{2} \left(  \sum_{i=1}^d \left( \log \frac{\lambda_i(\Sigma_2)}{\lambda_i(\Sigma_1)}+\lambda_i(\Sigma_2^{-1}\Sigma_1) \right) -d\right),\\
&=&\frac{1}{2} \sum_{i=1}^d \left(\lambda_i(\Sigma_2^{-1}\Sigma_1)  -\log\lambda_i(\Sigma_2^{-1}\Sigma_1)-1\right),\\
&=& (f_\KL\circ\Lambda)(\Sigma_2^{-1}\Sigma_1),
\end{eqnarray}
where  $|\Sigma|$ and $\lambda_i(\Sigma)$ denotes respectively the determinant of a positive-definite matrix $\Sigma\succ 0$, 
and the $i$-the
  real largest eigenvalue of $\Sigma$, and 
\begin{equation}
p_\Sigma(x)=\frac{1}{\sqrt{(2\pi)^d |\Sigma|}}\exp\left(-\frac{1}{2} x^\top\Sigma^{-1} x\right)
\end{equation}
is the density of the multivariate zero-centered Gaussian of covariance matrix $\Sigma$,  
\begin{equation}
f_\KL(u_1,\ldots,u_d)=\frac{1}{2}\left(\sum_{i=1}^d (u_i-1-\log u_i)\right),
\end{equation}
is a
 symmetric function  invariant under parameter permutations,
and $\Lambda(\cdot)$ denotes the eigenvalue map.

This Siegel distance in the upper plane is also a {\em smooth} spectral distance function since we have
\begin{equation}
\rho_U(Z_1,Z_2)=f\circ \Lambda(R(Z_1,Z_2)),
\end{equation}
where  $f$ is the following
 symmetric function: 
\begin{equation}
f(x_1,\ldots, x_d)=\sqrt{\sum_{i=1}^d \log^2\left(\frac{1+\sqrt{x_i}}{1-\sqrt{x_i}}\right)}.
\end{equation}

A remarkable property is that all eigenvalues of $R(Z_1,Z_2)$ are positive (see~\cite{Siegel-1943}) although $R$ may not necessarily be a Hermitian matrix. In practice, when calculating {\em numerically} the eigenvalues of the complex matrix $R(Z_1,Z_2)$, we  obtain very small imaginary parts which shall be rounded to zero.
Thus calculating  the Siegel distance on the upper plane requires cubic time, i.e., the cost of computing the eigenvalue decomposition.

This Siegel distance in the upper plane $\bbSH(d)$ generalizes several well-known distances:

\begin{itemize}

\item  When $Z_1=iY_1$ and $Z_2=iY_2$, we have 
\begin{equation}
\rho_U(Z_1,Z_2)=\rho_\PD(Y_1,Y_2),
\end{equation}
the Riemannian distance between $Y_1$ and $Y_2$ on the symmetric positive-definite manifold~\cite{forstner2003metric,Moakher-2005}:
\begin{eqnarray}
\rho_\PD(Y_1,Y_2) &=& \|\Log(Y_1Y_2^{-1})\|_F\\
&=& \sqrt{\sum_{i=1}^d \log^2\left(\lambda_i(Y_1Y_2^{-1})\right)}.
\end{eqnarray}

In that case, the Siegel upper metric for $Z=iY$ becomes the affine-invariant metric:
\begin{equation}
\ds_U^2(Z) =  \tr\left((Y^{-1}\dY)^2\right) = \ds_\PD(Y).
\end{equation}

Indeed, we have $\rho_\PD(C^\top Y_1C,C^\top Y_2C)=\rho_\PD(Y_1,Y_2)$ for any $C\in\GL(d,\bbR)$ and 
\begin{equation}
\rho_\PD(Y_1^{-1},Y_2^{-1})=\rho_\PD(Y_1,Y_2).
\end{equation}

\item  In 1D, the Siegel upper distance $\rho_{U}(Z_1,Z_2)$ between $Z_1=[z_1]$ and $Z_2=[z_2]$ (with $z_1$ and $z_2$ in $\bbC$) amounts to the hyperbolic distance on the Poincar\'e upper plane $\bbH$:
\begin{equation}
\rho_{U}(Z_1,Z_2)=\rho_{U}(z_1,z_2),
\end{equation}
where
\begin{equation}\label{eq:PoincareDistLog}
\rho_{U}(z_1,z_2) := \log \frac{|z_1-\barz_2|+|z_1-z_2| }{|z_1-\barz_2|-|z_1-z_2|}.
\end{equation}

\item
The Siegel distance between two diagonal matrices $Z=\diag(z_1,\ldots,z_d)$ and $Z'=\diag(z_1',\ldots,z_d')$ is
\begin{equation}
\rho_{U}(Z,Z')= \sqrt{\sum_{i=1}^d  \rho_{U}^2(z_i,z_i')}.
\end{equation}
Observe that the Siegel distance is a non-separable metric distance, but its squared distance is separable when the matrices are diagonal:
\begin{equation}
\rho_{U}^2(Z,Z')=  \sum_{i=1}^d  \rho_{U}^2(z_i,z_i').
\end{equation}
\end{itemize}

The Siegel  metric in the upper plane is invariant by generalized matrix M\"obius transformations (linear fractional transformations or rational transformations):
\begin{equation}
\phi_S(Z)  := (AZ+B)(CZ+D)^{-1},
\end{equation}
where $S\in M(2d,\bbR)$ is the following $2d\times 2d$ block matrix:
\begin{equation}
S=\mattwotwo{A}{B}{C}{D}.
\end{equation}
which satisfies
\begin{equation}
AB^\top=BA^\top,\quad CD^\top=DC^\top,\quad AD^\top-BC^\top=I.
\end{equation}
The map $\phi_S(\cdot)= \phi(S,\cdot)$ is called a {\em symplectic map}.

The set of  matrices $S$ encoding the symplectic maps forms a group called the {\em real symplectic group} $\Sp(d,\bbR)$~\cite{Freitas-1999} (informally, the group of Siegel motions):
\begin{equation}
\Sp(d,\bbR)=\left\{
 \mattwotwo{A}{B}{C}{D},\quad A,B,C,D\in\ M(d,\bbR) : AB^\top=BA^\top,\quad CD^\top=DC^\top,\quad AD^\top-BC^\top=I
\right\}.
\end{equation}
It can be shown that symplectic matrices have unit determinant~\cite{mackey2003determinant,rim2017elementary}, and therefore  $\Sp(d,\bbR)$ is a subgroup of $\SL({2d},\bbR)$,   the special group of  real invertible  matrices with unit determinant. 
We also check that if $M\in\Sp(d,\bbR)$ then $M^\top\in\Sp(d,\bbR)$.

Matrix $S$ denotes the representation of the group element $g_S$.
The symplectic group operation corresponds to matrix multiplications of their representations,  the neutral element is encoded by $E=\mattwotwo{I}{0}{0}{I}$,
and the group inverse of $g_S$ with $S=\mattwotwo{A}{B}{C}{D}$ is encoded by the matrix:
\begin{equation}
S^{(-1)}\eqdef  \mattwotwo{D^\top}{-B^\top}{-C^\top}{A^\top}.
\end{equation}
Here, we use the parenthesis notation $S^{(-1)}$ to indicate that it is the {\em group inverse} and not the usual matrix inverse $S^{-1}$.
The symplectic group is a Lie group of dimension $d(2d + 1)$.
Indeed, a symplectic matrix of $\Sp(d,\bbR)$ has $2d\times 2d=4d^2$ elements which are constrained from the block matrices as follows:
\begin{eqnarray}
AB^\top&=&BA^\top,\\
CD^\top&=&DC^\top,\\
AD^\top-BC^\top&=&I.
\end{eqnarray}
The first two constraints are independent and of the form $M=M^\top$ which yields each $\frac{d^2-d}{2}$ elementary constraints.
The third constraint is of the form $M_1-M_2=I$, and independent of the other constraints, yielding $d^2$ elementary constraints.
Thus the dimension of the symplectic group is 
\begin{equation}
\dim(\Sp(d,\bbR))=4d^2-(d^2-d)-d^2=2d^2+d=d(2d+1).
\end{equation}

The {\em action} of the group is {\em transitive}:
That is,  for any $Z=A+iB$ and $S(Z)=\mattwotwo{B^{-\frac{1}{2}} }{0}{AB^{-\frac{1}{2}}}{B^{\frac{1}{2}}}$, we have
$\phi_{S(Z)}(iI)=Z$.
Therefore, by taking the group inverse
\begin{equation}
S^{(-1)}=\mattwotwo{(B^{\frac{1}{2}})^\top}{0}{-(AB^{-\frac{1}{2}})^\top}{(B^{-\frac{1}{2}})^\top},
\end{equation}
we get
\begin{equation}
\phi_{S^{(-1)}}(Z)=iI.
\end{equation}
The action $\phi_S(Z)$ can be interpreted as a ``Siegel translation'' moving  matrix $iI$ to matrix $Z$, 
and conversely the action $\phi_{S^{(-1)}(Z)}$ as moving matrix $Z$ to matrix $iI$.

The {\em stabilizer group} of $Z=iI$ (also called isotropy group, the set of group elements $S\in\Sp(d,\bbR)$ whose action fixes $Z$) is the subgroup of {\em symplectic orthogonal matrices} $\SpO(2d,\bbR)$:
\begin{equation}
\SpO(2d,\bbR) = \left\{ \mattwotwo{A}{B}{-B}{A} \ :\ A^\top A+B^\top B=I, A^\top B\in\Sym(d,\bbR) \right\}.
\end{equation}
We have $\SpO(2d,\bbR) =\Sp(2d,\bbR)\cap O(2d)$, where $O(2d)$ is the {\em group of orthogonal matrices} of dimension  $2d\times 2d$:
\begin{equation}
O(2d) :=\left\{ R\in M(2d,\bbR)\ :\ RR^\top=R^\top R=I  \right\}.
\end{equation}
Informally speaking, the elements of $\SpO(2d,\bbR)$ represent the ``Siegel rotations'' in the  upper plane.
The Siegel upper plane is isomorphic to $\Sp(2d,\bbR)/O_d(\bbR)$.

A pair of matrices $(Z_1,Z_2)$ can be transformed into another pair of matrices $(Z_1',Z_2')$ of $\SH(d)$ if and only if
 $\lambda(R(Z_1,Z_2))=\lambda(R(Z_1',Z_2'))$, where $\lambda(M):=\{\lambda_1(M),\ldots, \lambda_d(M)\}$ denotes the {\em spectrum} of matrix $M$.

By noticing that the symplectic group elements $M$ and $-M$ yield the same symplectic map, we define
the orientation-preserving isometry group of the Siegel upper plane as the {\em real projective symplectic group} 
$\PSp(d,\bbR)=\Sp(d,\bbR)/\{\pm I_{2d}\}$ 
(generalizing the group $\PSL(2,\bbR)$ obtained when $d=1$).

The geodesics in the Siegel upper space can be obtained by applying symplectic transformations to the geodesics of the {\em positive-definite manifold} (geodesics on the SPD manifold) which is a totally geodesic submanifold of $\SU(d)$.
Let $Z_1=iP_1$ and $Z_2=iP_2$.
Then the geodesic $Z_{12}(t)$ with $Z_{12}(0)=Z_1$ and $Z_{12}(1)=Z_2$ is expressed as:
\begin{equation}
Z_{12}(t)=i P_1^{\frac{1}{2}} \Exp(t\ \Log(P_1^{-\frac{1}{2}} P_2 P_1^{-\frac{1}{2}})) P_1^{\frac{1}{2}},
\end{equation}
where $\Exp(M)$ denotes the {\em matrix exponential}:
\begin{equation}
\Exp(M)=\sum_{i=0}^\infty \frac{1}{i!}M^i,
\end{equation}
and $\Log(M)$ is the  {\em principal matrix logarithm}, unique when matrix $M$ has all positive eigenvalues.

The equation of the geodesic emanating from $P$ with tangent vector $S\in T_p$ (symmetric matrix) on the SPD manifold  is: 
\begin{equation}
\gamma{P,S}(t)=P^{\frac{1}{2}} \Exp(t P^{-\frac{1}{2}} S P^{-\frac{1}{2}} ) P^{\frac{1}{2}}.
\end{equation}

Both the exponential and the principal logarithm of a matrix $M$ can be calculated in cubic time when the matrices are diagonalizable:
Let $V$ denote the matrix of eigenvectors so that we have the following decomposition: 
\begin{equation}
M=V\ \diag(\lambda_1,\ldots,\lambda_d)\ V^{-1},
\end{equation}
 where $\lambda_1, \ldots, \lambda_d$ are the corresponding eigenvalues of eigenvectors.
Then for a scalar function $f$ (e.g., $f(u)=\exp(u)$ or $f(u)=\log u$), we define the corresponding matrix function $f(M)$ as
\begin{equation}
f(M) := V\ \diag(f(\lambda_1),\ldots,f(\lambda_d))\ V^{-1}.
\end{equation}

The volume element of the Siegel upper plane is $2^{\frac{d(d-1)}{2}}\dv$ where $\dv$ is the volume element of the $d(d+1)$-dimensional Euclidean space expressed in the Cartesian coordinate system.

\section{The Siegel disk domain and the Kobayashi distance}\label{sec:SiegelDisk}

The {\em Siegel disk}~\cite{Siegel-1943}  is an open convex complex matrix domain defined by
\begin{eqnarray}
\SD(d) &:=& \left\{  W\in\Sym(d,\bbC) \st I-\barW W\succ  0\right\}.
\end{eqnarray}
The Siegel disk can be written equivalently as $\SD(d) := \left\{  W\in\Sym(d,\bbC) \st I-W\barW\succ  0\right\}$ or 
$\SD(d) := \left\{  W\in\Sym(d,\bbC) \st  \|W\|_O <1\right\}$.
In the Cartan classification~\cite{Cartan-1935}, the Siegel disk is a Siegel domain of type III.

When $d=1$, the Siegel disk $\SD(1)$ coincides with the Poincar\'e disk: $\SD(1)=\bbD$.
The Siegel disk was described by Siegel~\cite{Siegel-1943} (page 2, called domain E to contrast with domain H of the upper space) and Hua in his 1948's paper~\cite{hua1947geometries2} (page 205) on the geometries of matrices~\cite{wan1996geometry}. Siegel's paper~\cite{Siegel-1943} in 1943 only considered the Siegel upper plane.
Here, the Siegel (complex matrix) disk is not to be confused with the other notion of Siegel disk in complex dynamics which is a connected component in the Fatou set.

The boundary $\partial\SD(d)$ of the Siegel disk  is called the {\em Shilov boundary}~\cite{clerc2009geometry,Freitas-1999,freitas2004revisiting}): $\partial\SD(d) := \left\{  W\in\Sym(d,\bbC) \st  \|W\|_O =1\right\}$. 
We have $\partial\SD(d)=\Sym(d,\bbC)\cap U(d,\bbC)$, where 
\begin{equation}
U(d,\bbC)=\{UU^*=U^*U=I\ : \ U\in M(d,\bbC)\}
\end{equation}
 is the group of $d\times d$ unitary matrices.
Thus $\partial\SD(d)$ is the set of {\em symmetric}  $d\times d$ unitary matrices with determinant of unit module.
The Shilov boundary is a {\em stratified manifold} where each stratum is defined as a space of constant rank-deficient matrices~\cite{bassanelli1983horospheres}.

The metric in the Siegel disk is: 
\begin{equation}
\ds_D^2 = \tr \left( (I-W\barW)^{-1} \dW (I-W\barW)^{-1}\dbarW \right).
\end{equation}
When $d=1$, we  recover $\ds_D^2=\frac{1}{(1-|w|^2)^2}\dw\dbarw$ which is the usual metric in the Poincar\'e disk (up to a missing factor of $4$, see Eq.~\ref{eq:PoincareDiskMetric}).

This Siegel metric induces a {\em K\"ahler geometry}~\cite{Barbaresco-2011} with the following {\em K\"ahler potential}:
\begin{equation}
K(W) = -\tr\left(\Log \left( I-W^HW\right)\right).
\end{equation}

The {\em Kobayashi distance}~\cite{Kobayashi-1967} between $W_1$ and $W_2$ in $\SD(d)$ is calculated~\cite{bassanelli-1983} as follows:
\begin{equation}\label{eq:DistSiegelDisk}
\rho_D(W_1,W_2) = \log\left( \frac{1+\|\Phi_{W_1}(W_2)\|_O}{1-\|\Phi_{W_1}(W_2)\|_O} \right),
\end{equation}
where
\begin{equation}\label{eq:translationSiegelDisk}
\Phi_{W_1}(W_2)=(I-W_1\barW_1)^{-\frac{1}{2}} (W_2-W_1) (I-\barW_1W_2)^{-1} (I-\barW_1W_1)^{\frac{1}{2}},
\end{equation}
is a Siegel translation which moves $W_1$ to the origin $O$ (matrix with all entries set to $0$) of the disk: 
We have $\Phi_{W}(W)=0$.
In the Siegel disk domain, the Kobayashi distance~\cite{Kobayashi-1967} coincides with the Carath\'eodory distance~\cite{Caratheodory-1927} and yields a metric distance.
Notice that the Siegel disk distance, although a spectral distance function via the operator norm, is {\em not smooth} because of it uses the maximum singular value. Recall that the Siegel upper plane distance uses {\em all} eigenvalues of a matrix cross-ratio $R$.

It follows that the cost of calculating a Kobayashi distance in the Siegel disk is cubic:
We require to compute a {\em symmetric matrix square root}~\cite{sra2015matrix} in Eq.~\ref{eq:translationSiegelDisk}, and
then compute the largest singular value for the operator norm in Eq.~\ref{eq:DistSiegelDisk}.

Notice that when $d=1$, the ``1d'' scalar matrices commute, and we have:
\begin{eqnarray}
\Phi_{w_1}(w_2) &=& (1-w_1\barw_1)^{-\frac{1}{2}} (w_2-w_1) (1-\barw_1w_2)^{-1} (1-\barw_1w_1)^{\frac{1}{2}},\\
&=& \frac{w_2-w_1}{1-\barw_1w_2}.
\end{eqnarray}
This corresponds to a hyperbolic translation of $w_1$ to $0$ (see Eq.~\ref{eq:transrot}).
Let us call the geometry of the Siegel disk the Siegel-Poincar\'e geometry.

We observe the following special cases of the Siegel-Poincar\'e distance:
\begin{itemize}

\item Distance to the origin: When $W_1=0$ and $W_2=W$, we have $\Phi_0(W)=W$, and therefore the distance in the disk between a matrix $W$ and the origin $0$ is:
\begin{equation}\label{eq:distSiegelDiskO}
\rho_D(0,W) =  \log \left(\frac{1+\|W\|_O}{1-\|W\|_O}\right).  
\end{equation}
In particular, when $d=1$, we recover the formula of Eq.~\ref{eq:PoincareDiskOrigin}:
$\rho_D(0,w) = \log \left(\frac{1+|w|}{1-|w|}\right)$.

\item When $d=1$, we have $W_1=[w_1]$ and $W_2=[w_2]$, and
\begin{equation}
\rho_D(W_1,W_2) = \rho_D(w_1,w_2).
\end{equation}

\item Consider diagonal matrices $W=\diag(w_1,\ldots,w_d)\in\bbSD(d)$ and $W'=\diag(w_1',\ldots,w_d')\in\bbSD(d)$.
We have $|w_i|\leq 1$ for $i\in\{1,\ldots,d\}$. Thus the diagonal matrices belong to the {\em polydisk domain}.
Then we have
\begin{equation}
\rho_D(W_1,W_2)=\sqrt{\sum_{i=1}^d \rho_D^2(w_i,w_i')}.
\end{equation}
Notice that the polydisk domain is a Cartesian product of 1D complex disk domains, but it is {\em not} the {\em unit $d$-dimensional complex ball}
$\{z\in\mathbb{C}^d \ :\ \sum_{i=1^d} z_i\bar{z}_i=1 \}$. 
\end{itemize}

We can convert a matrix $Z$ in the Siegel upper space to an equivalent matrix $W$ in the Siegel disk  
 by using the following  {\em matrix Cayley transformation} for $Z\in \SH_d$:
\begin{equation}
W_{U\rightarrow D}(Z):=(Z-iI)(Z+iI)^{-1}\in \SD(d).
\end{equation}

Notice that the imaginary positive-definite matrices $iP$ of the upper plane (vertical axis) are mapped to  
\begin{equation}
W_{U\rightarrow D}(iP):=(P-I)(P+I)^{-1}\in \SD(d),
\end{equation}
i.e., the real symmetric matrices belonging to the horizontal-axis of the disk.

The inverse transformation for a matrix $W$ in the Siegel disk   is
\begin{equation}
Z_{D\rightarrow U}(W)=i\left(I+W\right)\left(I-W\right)^{-1}\in\SH(d),
\end{equation}
a matrix in the Siegel upper space.
With those mappings, the origin of the disk $0\in\SD(d)$ coincides with matrix $iI\in\SH(d)$ in the upper space.

A key property is that the geodesics passing through the matrix origin $0$ are expressed by {\em straight line segments} in the Siegel disk.
We can check that 
\begin{equation}
\rho_D(0,W)=\rho_D(0,\alpha W)+\rho_D(\alpha W,W),
\end{equation}
 for any $\alpha\in [0,1]$.

To describe the geodesics between $W_1$ and $W_2$, we first move $W_1$ to $0$ and $W_2$ to $\Phi_{W_1}(W_2)$.
Then the geodesic between $0$ and $\Phi_{W_1}(W_2)$ is a straight line segment, and we map back this geodesic via $\Phi^{-1}{W_1}(\cdot)$.
The inverse of a symplectic map is a symplectic map which corresponds to the action of an element of the complex symplectic group.

The {\em complex symplectic group} is
\begin{equation}
\Sp(d,\bbC)= \left\{ M^\top J M =J,  M=\mattwotwo{A}{B}{C}{D}\in M(2d,\bbC)   \right\},
\end{equation}
 with 
\begin{equation}
J=\mattwotwo{0}{I}{-I}{0},
\end{equation}
for the $d\times d$ identity matrix  $I$.
Notice that the condition $M^\top J M = J$ amounts to check that
\begin{equation}
AB^\top=BA^\top,\quad CD^\top=DC^\top,\quad AD^\top-BC^\top=I.
\end{equation}

The conversions between the Siegel upper plan to the Siegel disk (and vice versa) can be expressed using {\em complex symplectic transformations} associated to the matrices:
\begin{eqnarray}
W(Z)=\mattwotwo{I}{-iI}{I}{iI}.Z=(Z-iI)(Z+iI))^{-1},\\
Z(W)=\mattwotwo{iI}{iI}{-I}{I}.W=i\left(I+W\right)\left(I-W\right)^{-1}.
\end{eqnarray}

Figure~\ref{fig:ConversionUpperDisk} depicts the conversion of the upper plane to the disk, and vice versa.

\begin{figure}
\centering
\includegraphics[width=\columnwidth]{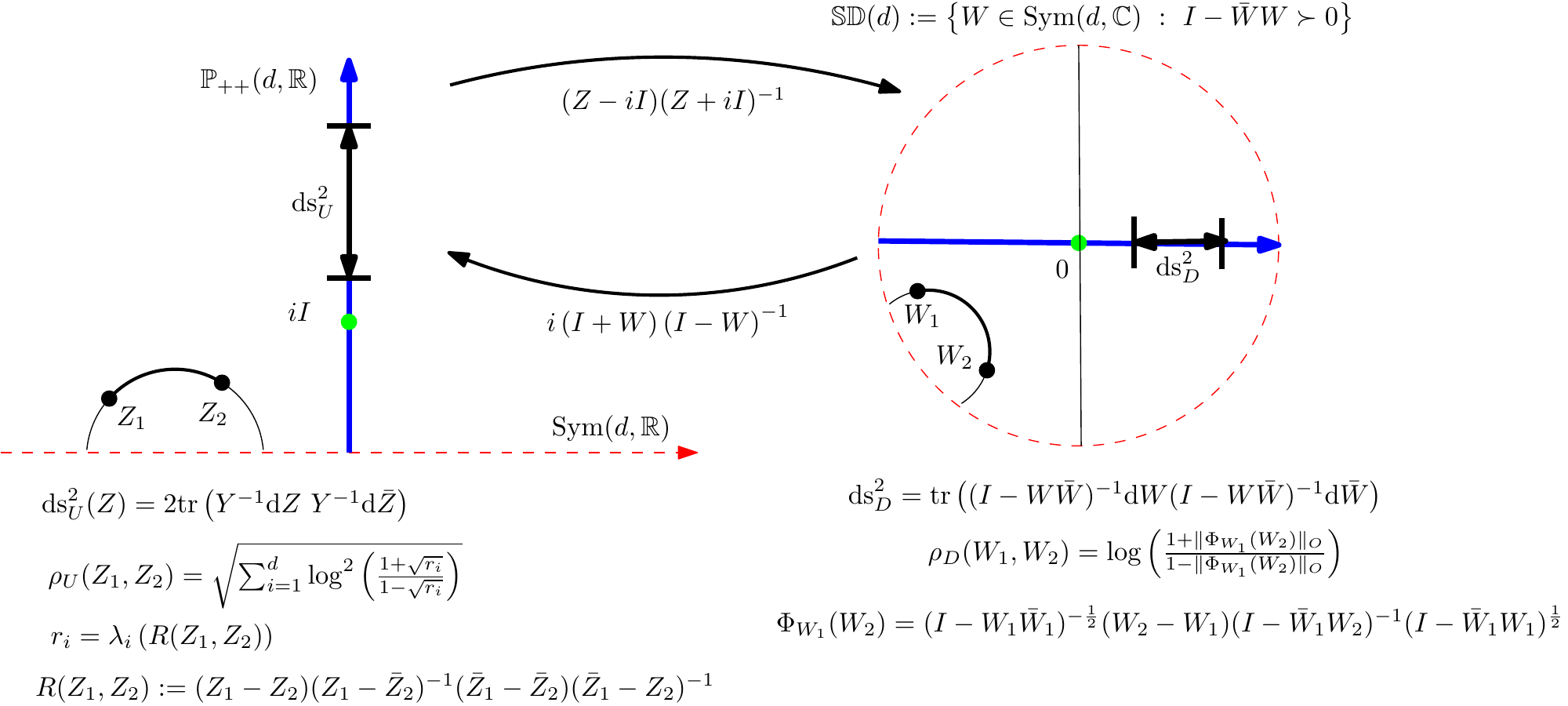}

\caption{Illustrating the properties and conversion between the Siegel upper plane and the Siegel disk.\label{fig:ConversionUpperDisk}}
\end{figure}

The orientation-preserving isometries in the Siegel disk is the {\em projective complex symplectic group}
 $\PSp(d,\bbC)=\Sp(d,\bbC)/\{\pm I_{2d}\}$.

It can be shown that
\begin{equation}
\Sp(d,\bbC)= \left\{   M=\mattwotwo{A}{B}{\bar{B}}{\bar{A}}\in M(2d,\bbC)   \right\},
\end{equation}
with
\begin{eqnarray}
A^\top\bar{B}-B^H A&=&0,\\
A^\top\bar{A}-B^HB&=&I.
\end{eqnarray}
and the left action of $g\in\Sp(d,\bbC)$ is
\begin{equation}
g.W= (AW+B)(\bar{A}W+\bar{B})^{-1}.
\end{equation}

The {\em isotropy group} at the origin $0$ is 
\begin{equation}
\left\{ \mattwotwo{A}{0}{0}{\bar{A}} \ :\ A\in U(d) \right\},
\end{equation}
where $U(d)$ is the {\em unitary group}: $U(d)=\{ U\in \GL(d,\bbC) \ :\ U^HU=UU^H=I\}$.

Thus we can ``rotate'' a matrix $W$ with respect to the origin so that its imaginary part becomes $0$:
There exists $A$ such that $\Re(AWW^{-1}\bar{A}^{-1})=0$.

More generally, we can define a Siegel rotation~\cite{mitchell1955potential} in the disk with respect to a center $W_0\in\SD(d)$ as follows:
\begin{equation}
R_{W_0}(W) = (AW-AW_0)(B-B\barW_0W)^{-1},
\end{equation}
where
\begin{eqnarray}
\bar{A}A &=& (I-W_0\barW_0)^{-1},\\
\bar{B}B &=& (I-\barW_0W_0)^{-1},\\
\bar{A}{A}W_0&=&W_0\bar{B}B.
\end{eqnarray}

Interestingly, the Poincar\'e disk can be embedded {\em non-diagonally} onto the Siegel upper plane~\cite{rong2018non}.

In complex dimension $d=1$, the Kobayashi distance $\rho_W$ coincides with the Siegel distance $\rho_U$.
Otherwise, we calculate the Siegel distance in the Siegel disk as
\begin{equation}
\rho_U(W_1,W_2):=\rho_U\left(Z_{D\rightarrow U}(W_1),Z_{D\rightarrow U}(W_2)\right).
\end{equation}

\section{The Siegel-Klein geometry: Distance and geodesics}\label{sec:SiegelKlein}

We define the Siegel-Klein geometry as the Hilbert geometry for the Siegel disk model.
Section~\ref{sec:HG} concisely explains the Hilbert geometry induced by an open bounded convex domain.
In \S\ref{sec:HGSiegelDisk}, we study the Hilbert geometry of the Siegel disk domain.
Then we report the Siegel-Klein distance in~\S\ref{sec:SKdistance} and study some of its particular cases.
Section~\S\ref{sec:SiegelConversions} presents the conversion procedures between the Siegel-Poincar\'e disk and the Siegel-Klein disk.
In~\S\ref{sec:SKdistapprox}, we design a fast guaranteed method to approximate the Siegel-Klein distance.
Finally, we introduce the Hilbert-Fr\"obenius distances to get simple bounds on the Siegel-Klein distance in \S\ref{sec:HF}.

\subsection{Background on Hilbert geometry}\label{sec:HG}

\begin{figure}
\centering
\includegraphics[width=0.35\columnwidth]{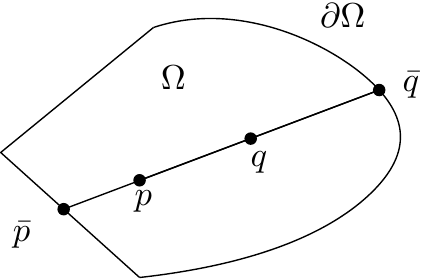}

\caption{Hilbert distance induced by a bounded open convex domain $\Omega$.\label{fig:HilbertOmega}}
\end{figure}

Consider a normed vector space $(V,\|\cdot\|)$, and define the Hilbert distance~\cite{Hilbert-1895,beardon1999klein} for an open bounded convex domain $\Omega$ as follows:

\begin{Definition}[Hilbert distance]
The Hilbert distance is defined for any open bounded convex domain $\Omega$ and a prescribed positive factor $\kappa>0$ by 
\begin{equation}
H_{\Omega,\kappa}(p,q) :=  \left\{
 \begin{array}{ll}
 \kappa \log \left|\CR(\bar{p},p;q,\bar{q})\right|, & p\not=q,\\
 0 & p=q.
\end{array}
\right.
\end{equation}
where $\bar{p}$ and $\bar{q}$ are the unique two intersection points of the line $(pq)$ 
with the  boundary $\partial\Omega$ of the domain $\Omega$ as depicted in Figure~\ref{fig:HilbertOmega}, and  $\CR$ denotes the cross-ratio of four points (a projective invariant):
\begin{equation}
\CR(a,b;c,d) = \frac{\|a-c\| \|b-d\|}{\|a-d\| \|b-c\|}.
\end{equation}
\end{Definition}

When $p\not=q$, we have:
\begin{equation}
H_{\Omega,\kappa}(p,q) :=  \kappa \log \left(\frac{\|\bar{q}-p\| \|\bar{p}-q \|}{\|\bar{q}-q\| \|\bar{p}-p\|}\right).
\end{equation}

The Hilbert distance is a {\em metric distance} which does {\em not} depend on the underlying norm of the vector space:

\begin{Proposition}[Formula of Hilbert distance]\label{prop:Hilberformula}
The Hilbert distance between two points $p$ and $q$ of an open bounded convex domain $\Omega$ is
\begin{equation}
H_{\Omega,\kappa}(p,q) =\left\{
\begin{array}{ll}
\kappa \log \left|\frac{\alpha_+(1-\alpha_-)}{\alpha_-(\alpha_+-1)}\right|, & p\not=q,\\
0 & p=q.
\end{array}
\right.,
\end{equation}
where $\bar{p}=p+\alpha^-(q-p)$ and $\bar{q}=p+\alpha^+(q-p)$ are the two intersection points of the line $(pq)$ with the boundary $\partial\Omega$ of the domain $\Omega$.
\end{Proposition}

\begin{proof}
For distinct points $p$ and $q$ of $\Omega$, let $\alpha^+>1$ be such that $\bar{q}=p+\alpha^+(q-p)$, and $\alpha_-<0$ such that
$\bar{p}=p+\alpha^-(q-p)$.
Then we have
$\|\bar{q}-p\|=\alpha_+\|q-p\|$, $\|\bar{p}-p\|=|\alpha_-| \|q-p\|$, $\|q-\bar{q}\|=(\alpha_+-1)\|p-q\|$ and $\|\bar{p}-q\|=(1-\alpha_)\|p-q\|$.
Thus we get
\begin{eqnarray}
H_{\Omega,\kappa}(p,q) &=&  \kappa \log \frac{\|\bar{q}-p\| \|\bar{p}-q \|}{\|\bar{q}-q\| \|\bar{p}-p\|},\\
&=& \kappa \log\left(\frac{\alpha_+(1-\alpha_-)}{|\alpha_-|(\alpha_+-1)}\right),
\end{eqnarray}
and $H_\Omega(p,q)=0$ if and only if $p=q$.
\end{proof}

We may also write the source points $p$ and $q$ as linear interpolations of the extremal points $\bar{p}$ and $\bar{q}$ on the boundary:
$p=(1-\beta_p)\bar{p}+\beta_p\bar{q}$ and $q=(1-\beta_q)\bar{p}+\beta_q\bar{q}$ with $0<\beta_p<\beta_q<1$ for distinct points $p$ and $q$. 
In that case, the Hilbert distance can be written as
\begin{eqnarray}\label{eq:HilbertDist1D}
H_{\Omega,\kappa}(p,q) =\left\{
\begin{array}{ll}
\kappa \log \left(\frac{1-\beta_p}{\beta_p} \frac{\beta_q}{1-\beta_q}\right) & \beta_p\not=\beta_q,\\
0 & \beta_p=\beta_q.
\end{array}
\right.
\end{eqnarray}

The projective Hilbert space $(\Omega,H_\Omega)$ is a metric space.
Notice that the above formula has demonstrated that 
\begin{equation}
H_{\Omega,\kappa}(p,q)=H_{\Omega\cap (pq),\kappa}(p,q).
\end{equation}
That is, the Hilbert distance between two points of a $d$-dimensional domain $\Omega$ is equivalent to the Hilbert distance between the two points on the 1D domain $\Omega\cap (pq)$ defined by $\Omega$ restricted to the line $(pq)$ passing through the points $p$ and $q$.

Notice that the boundary $\partial\Omega$ of the domain may not be smooth (e.g., $\Omega$ may be a simplex~\cite{nielsen2018clustering} or a polytope~\cite{nielsen2017balls}).
The Hilbert geometry for the unit disk centered at the origin with $\kappa=\frac{1}{2}$ yields the Klein model~\cite{Klein-1873} 
(or Klein-Beltrami model~\cite{Beltrami-1868}) of hyperbolic geometry.
The Hilbert geometry for an ellipsoid   yields the {\em Cayley-Klein hyperbolic model}~\cite{Cayley-1859,richter2011perspectives,CKclassification-2016} generalizing the Klein model.
The Hilbert geometry for a simplicial polytope is isometric to a normed vector space~\cite{de1993hilbert,nielsen2018clustering}.
We refer to the handbook~\cite{papadopoulos2014handbook} for a survey of recent results on Hilbert geometry.
The Hilbert geometry of the elliptope (i.e., space of correlation matrices) was studied in~\cite{nielsen2018clustering}.
Hilbert geometry may be studied from the viewpoint of {\em Finslerian geometry} which is Riemannian if and only if the domain $\Omega$ is an ellipsoid (i.e., Klein or Cayley-Klein hyperbolic geometries).
Last but not least, it is interesting to observe the similarity of the Hilbert distance which relies on a geometric cross-ratio with  the 
Siegel distance (Eq.~\ref{eq:SiegelDistance}) in the upper space  which relies on a matrix generalization 
of the cross-ratio (Eq.~\ref{eq:matrixcr}).

\subsection{Hilbert geometry of the Siegel disk domain}\label{sec:HGSiegelDisk}

\begin{figure}
\centering
\includegraphics[width=0.35\columnwidth]{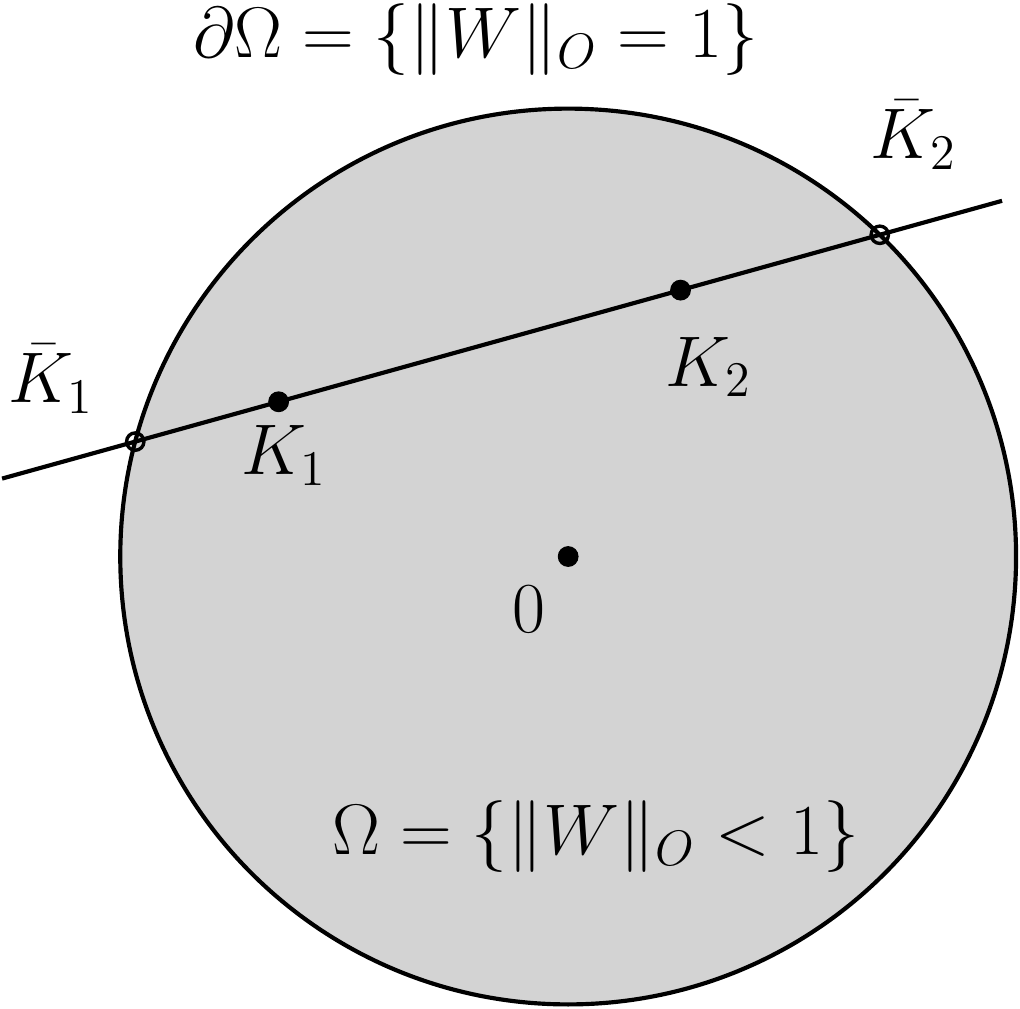}

\caption{Hilbert geometry for the Siegel disk: The Siegel-Klein disk model.\label{fig:HilbertKlein}}
\end{figure}

Let us consider the Siegel-Klein disk model which is defined as the Hilbert geometry for the Siegel disk domain $\Omega=\SD(d)$ as depicted in Figure~\ref{fig:HilbertKlein} with $\kappa=\frac{1}{2}$.

\begin{Definition}[Siegel-Klein geometry]\label{def:SiegelKlein}
The Siegel-Klein disk model is the Hilbert geometry for the  open bounded convex domain $\Omega=\SD(d)$ with prescribed constant $\kappa=\frac{1}{2}$.
The Siegel-Klein distance is 
\begin{equation}
\rho_K(K_1,K_2):=H_{\SD(d),\frac{1}{2}}(K_1,K_2).
\end{equation}
\end{Definition}

When $d=1$, the Siegel-Klein disk is the Klein disk model of hyperbolic geometry, and the Klein distance~\cite{HVD-2010} between two any points
 $k_1\in\bbC$ and $k_2\in\bbC$ restricted to the unit disk is
\begin{equation}
\rho_K(k_1,k_2) =  \mathrm{arccosh}\left(
\frac{1-(\Re(k_1)\Re(k_2)+\Im(k_1)\Im(k_2))}{\sqrt{(1-|k_1|)(1-|k_2|)}}
\right),
\end{equation}
where
\begin{equation}
\mathrm{arccosh}(x)=\log \left(x+\sqrt{x^{2}-1}\right), \quad x\geq 1.
\end{equation}

This formula can be retrieved from the Hilbert distance induced by the Klein unit disk~\cite{richter2011perspectives}.

\subsection{Calculating and approximating the Siegel-Klein distance}\label{sec:SKdistance}

The Siegel disk domain $\SD(d)= \left\{  W\in\Sym(d,\bbC) \st I-\barW W\succ  0\right\}$
 can be rewritten using the operator norm  as
\begin{equation}
\SD(d) = \left\{W\in\Sym(d,\bbC) \st \|W\|_O < 1 \right\}.
\end{equation}

Let $\{K_1+\alpha(K_2-K_1), \alpha\in\bbR\}$ denote the line passing through (matrix) points $K_1$ and $K_2$.
That line intersects the Shilov boundary when
\begin{equation}
\|K_1+\alpha(K_2-K_1)\|_O = 1.
\end{equation}

When $K_1\not=K_2$, there are two unique solutions since a line intersects the boundary of a bounded open convex domain in at most two points: 
Let one solution be $\alpha_+$ with $\alpha_+>1$, and the other solution be $\alpha_-$ with $\alpha_-<0$.
The Siegel-Klein distance is then defined as

\begin{equation}\label{eq:SiegelKleinDist}
\rho_K(K_1,K_2)= \frac{1}{2} \log \left( \frac{\alpha_+(1-\alpha_-)}{|\alpha_-|(\alpha_+-1)} \right),
\end{equation}
where $\barK_1=K_1+\alpha_-(K_2-K_1)$ and $\barK_2=K_1+\alpha_+(K_2-K_1)$ are the extremal matrices belonging to the Shilov boundary $\partial\SD(d)$.

Notice that matrices $K_1$ and/or $K_2$ may be rank deficient.
We have $\rank(K_1+\lambda(K_2-K_1))\leq \min(d, \rank(K_1)+\rank(K_2))$, see~\cite{marsaglia1964bounds}.

In practice, we may perform a bisection search on the matrix line $(K_1K_2)$ to approximate these two extremal points $\barK_1$ and $\barK_2$ 
(such that these matrices are ordered along the line as follows: $\barK_1$, $K_1$, $K_2$, $\barK_2$).
We may find a lower bound for $\alpha_-$ and a upper bound for  $\alpha_+$ as follows:
We seek $\alpha$ on the line $(K_1K_2)$ such that $K_1+\alpha(K_2-K_1)$ falls outside the Siegel disk domain:
\begin{equation}
1< \|K_1+\alpha(K_2-K_1)\|_O.
\end{equation}
Since $\|\cdot\|_O$ is a matrix norm, we have
\begin{equation}
1< \|K_1+\alpha(K_2-K_1)\|_O \leq \|K_1\|_O +|\alpha|\ \|(K_2-K_1)\|_O.
\end{equation}
Thus we deduce that
\begin{equation}
 |\alpha|  > \frac{1-\|K_1\|_O}{\|(K_2-K_1)\|_O}.
\end{equation}

\subsection{Siegel-Klein distance to the origin}\label{sec:SKdistanceO}

When $K_1=0$ (the $0$ matrix denoting the origin of the Siegel disk), and $K_2=K\in\SD(d)$, it is easy to solve the equation:
\begin{equation} 
\| \alpha K\|_O = 1.
\end{equation}
We have $|\alpha|=\frac{1}{\|K\|_O}$, that is, 
\begin{eqnarray}\label{eq:alphaorigin}
\alpha_+ &=& \frac{1}{\|K\|_O}>1, \\
\alpha_- &=& -\frac{1}{\|K\|_O}<0.
\end{eqnarray}
In that case, the Siegel-Klein distance of Eq.~\ref{eq:SiegelKleinDist} is expressed as:
\begin{eqnarray}
\rho_K(0,K)&=&\log \left( \frac{1+\frac{1}{\|K\|_O}}{\frac{1}{\|K\|_O}-1}\right),\\
&=&  \frac{1}{2} \log\left( \frac{1+\|K\|_O}{1-\|K\|_O}\right),\\
&=& 2\ \rho_D(0,K),
\end{eqnarray}
where $\rho_D(0,W)$ is defined in Eq.~\ref{eq:distSiegelDiskO}.

\begin{Theorem}[Siegel-Klein distance to the origin]\label{thm:SKdistorigin}
The Siegel-Klein distance of matrix $K\in\SD(d)$ to the origin $O$ is 
\begin{equation}\label{eq:sksp}
\rho_K(0,K)= \frac{1}{2} \log\left( \frac{1+\|K\|_O}{1-\|K\|_O}\right).
\end{equation}
\end{Theorem}
The constant $\kappa=\frac{1}{2}$ is chosen in order to ensure that when $d=1$ the corresponding Klein disk has negative {\em unit} curvature.
The result can be easily extended to the case of the Siegel-Klein distance between $K_1$ and $K_2$ where the origin $O$ belongs to the line $(K_1K_2)$.
In that case, $K_2=\lambda K_1$ for some $\lambda\in\bbR$ (e.g., $\lambda=\frac{\tr(K_2)}{\tr(K_1)}$ where $\tr$ denotes the matrix trace operator).
It follows that
\begin{eqnarray}
\|K_1+\alpha(K_2-K_1)\|_O &=& 1,\\
|1+\alpha(\lambda-1)| &=& \frac{1}{\|K_1\|_O}.
\end{eqnarray}
Thus we get the two values defining the intersection of $(K_1K_2)$ with the Shilov boundary:
\begin{eqnarray}
\alpha' &=& \frac{1}{\lambda-1} \left(\frac{1}{\|K_1\|_O}-1\right) ,\\
\alpha'' &=& \frac{1}{1-\lambda} \left(1+\frac{1}{\|K_1\|_O}\right).
\end{eqnarray}
We then apply formula Eq.~\ref{eq:SiegelKleinDist}:
\begin{eqnarray}\label{eq:SiegelKleinDistLineO}
\rho_K(K_1,K_2) &=& \frac{1}{2}  \left|\log\left( \frac{\alpha'(1-\alpha'')}{\alpha''(\alpha'-1)}  \right)\right|,\\
&=& \frac{1}{2}  \left|\log  \left( \frac{1-\|K_1\|_O}{1+\|K_1\|_O} \frac{\|K_1\|_O(1-\lambda)-(1+\|K_1\|_O)}{\|K_1\|_O(\lambda-1)-(1-\|K_1\|_O)} \right) \right|.
\end{eqnarray}

\begin{Theorem}
The Siegel-Klein distance between two points $K_1\not =0$ and $K_2$ on a line $(K_1K_2)$ passing through the origin is
$$
\rho_K(K_1,K_2)=\frac{1}{2}  \left|\log \left( \frac{1-\|K_1\|_O}{1+\|K_1\|_O} \frac{\|K_1\|_O(1-\lambda)-(1+\|K_1\|_O)}{\|K_1\|_O(\lambda-1)-(1-\|K_1\|_O)} \right) \right|,
$$
where $\lambda=\frac{\tr(K_2)}{\tr(K_1)}$.
\end{Theorem}

\subsection{Converting Siegel-Poincar\'e matrices from/to Siegel-Klein matrices}\label{sec:SiegelConversions}

From Eq.~\ref{eq:sksp}, we deduce that we can convert a matrix $K$ in the Siegel-Klein disk to a corresponding matrix $W$
 in the Siegel-Poincar\'e disk, 
and vice versa, as follows:

\begin{itemize}

\item Converting $K$ to $W$: 
We convert a matrix $K$ in the Siegel-Klein model to an equivalent matrix $W$ in the Siegel-Poincar\'e model as follows:
\begin{equation}
C_{K\rightarrow D}(K) = \frac{1}{1+\sqrt{1-\|K\|_O^2}}\ K.
\end{equation}
This conversion corresponds to a {\em radial contraction} with respect to the origin $0$ since $\frac{1}{1+\sqrt{1-\|K\|_O^2}}\leq 1$ (with equality for matrices belonging to the Shilov boundary). 

\item Converting $W$ to $K$: 
We convert a matrix $W$ in the Siegel-Poincar\'e model to an equivalent matrix $K$ in the Siegel-Klein model as follows:
\begin{equation}
C_{D\rightarrow K}(W) =  \frac{2}{1+\|W\|_O^2}\ W.
\end{equation}
This conversion corresponds to a {\em radial expansion} with respect to the origin $0$ since $\frac{2}{1+\|W\|_O^2}\geq 1$ (with equality for matrices on the Shilov boundary). 
\end{itemize}

\begin{Proposition}[Conversions Siegel-Poincar\'e$\Leftrightarrow$Siegel-Klein disk]\label{prop:convPK}
The conversion of a matrix $K$ of the Siegel-Klein model to its equivalent matrix $W$ in the Siegel-Poincar\'e model, and vice-versa, is done by the following radial contraction and expansion functions: 
$C_{K\rightarrow D}(K) = \frac{1}{1+\sqrt{1-\|K\|_O^2}}K$ and $C_{D\rightarrow K}(W) =   \frac{2}{1+|W|_O^2} W$.
\end{Proposition}

Figure~\ref{fig:conversions} illustrates the radial expansion/contraction conversions between the Siegel-Poincar\'e and Siegel-Klein matrices.

\begin{figure}
\centering
\includegraphics[width=0.85\columnwidth]{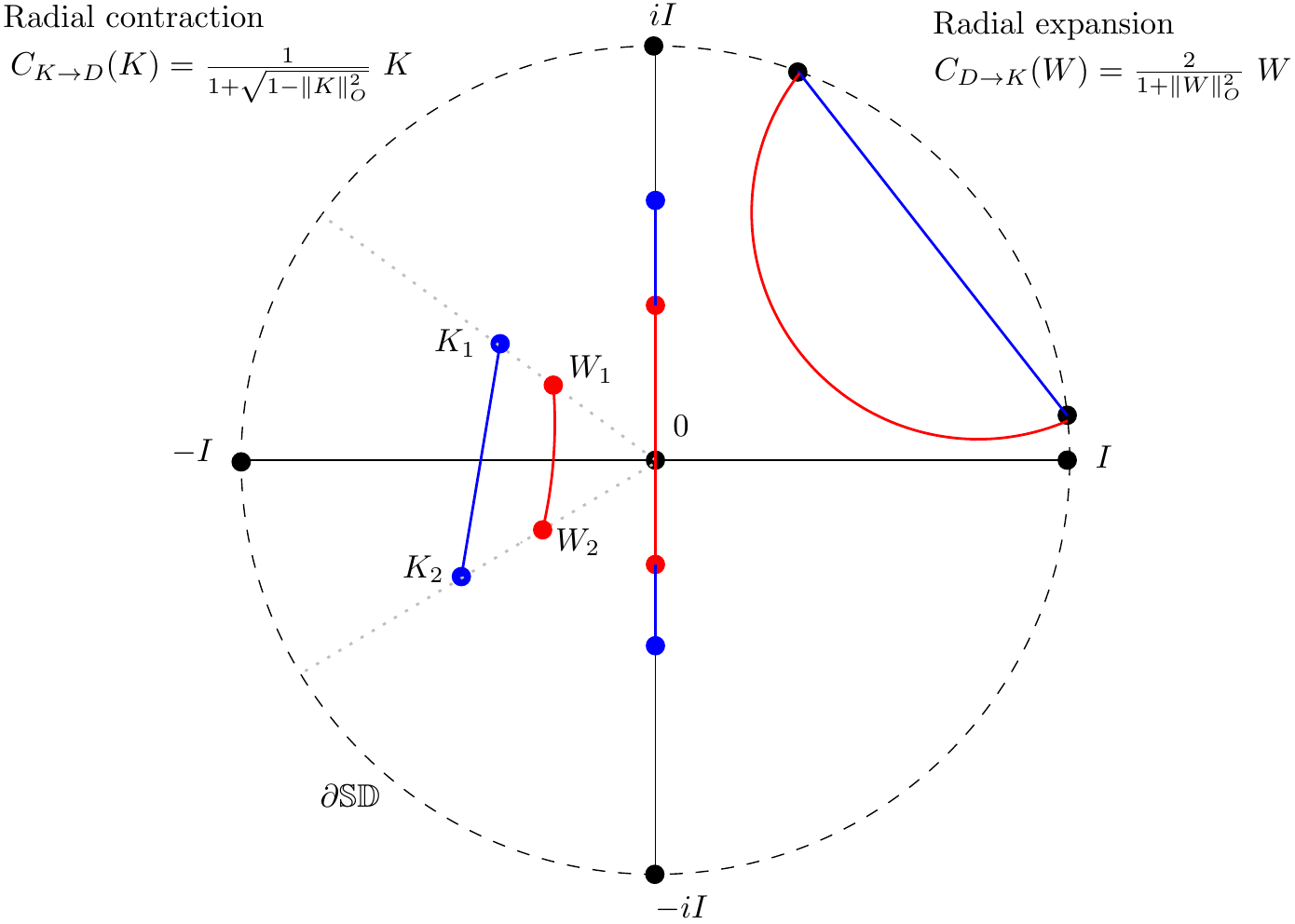}

\caption{Conversions in the Siegel disk domain: Poincar\'e to/from Klein matrices.\label{fig:conversions}}
\end{figure}

The cross-ratio
 $(p, q ; P, Q)=\frac{\|p-P\| \|q-Q\|}{\|p-Q\| \|q-P\|}$ of four collinear points on a line is such that 
$(p, q ; P, Q)=(p, r ; P, Q) \times(r, q ; P, Q)$ whenever $r$ belongs to that line.
By virtue of this cross-ratio property, the (pre)geodesics
 in the Hilbert-Klein disk are Euclidean straight.
Thus we can write the pregeodesics as:
\begin{equation}
\gamma_{K_1,K_2}(\alpha) = (1-\alpha)K_1+\alpha K_2 = K_1+\alpha(K_2-K_1).
\end{equation}

Riemannian geodesics are paths which minimize locally the distance and are parameterized proportionally to the arc-length. 
A pregeodesic is a path which minimizes locally the distance but is not necessarily parameterized proportionally to the arc-length.
For implementing geometric intersection algorithms (e.g., a geodesic with a ball), it is enough to consider pregeodesics.

Another way to get a generic closed-form formula for the Siegel-Klein distance is by using the formula for the Siegel-Poincar\'e disk after converting the matrices to their equivalent matrices in the Siegel-Poincar\'e disk.
We get the following expression:
\begin{eqnarray}\label{eq:SKdist}
\rho_K(K_1,K_2) &=&\rho_D(C_{K\rightarrow D}(K_1),C_{K\rightarrow D}(K_2)),\\
&=& \frac{1}{2}\log\left(\frac{1+\|\Phi_{C_{K\rightarrow D}(K_1)}(C_{K\rightarrow D}(K_2))\|_O}{1-\|\Phi_{C_{K\rightarrow D}(K_1)}(C_{K\rightarrow D}(K_2))\|_O}\right).
\end{eqnarray}

\begin{Theorem}[Formula for the Siegel-Klein distance]\label{thm:SKformulaViaP}
The Siegel-Klein distance between $K_1$ and $K_2$ in the Siegel disk is
$\rho_K(K_1,K_2)=\frac{1}{2}\log\left(\frac{1+\|\Phi_{C_{K\rightarrow D}(K_1)}(C_{K\rightarrow D}(K_2))\|_O}{1-\|\Phi_{C_{K\rightarrow D}(K_1)}(C_{K\rightarrow D}(K_2))\|_O}\right)$.
\end{Theorem}

The isometries in Hilbert geometry have been studied in~\cite{speer2014isometries}.

We now turn our attention to a special case where we can report an efficient and exact linear-time algorithm for calculating the Siegel-Klein distance.

\subsection{Siegel-Klein distance between diagonal matrices}\label{sec:SK:diagonal}
 Let $K_\alpha=K_1+\alpha K_{21}$ with $K_{21}=K_2-K_1$.
When solving for the general case, we seek for the extremal values of $\alpha$ such that:
\begin{eqnarray}
I-\overline{K}_\alpha K_\alpha &\succ& 0,\\
I-(\barK_1+\alpha \barK_{21})(K_1+\alpha K_{21}) &\succ& 0,\\
I-(\barK_1 K_1 +\alpha (\barK_1 K_{21} +\barK_{21}K_1)+\alpha^2 \barK_{21} K_{21}  ) &\succ& 0,\\
\barK_1 K_1 +\alpha (\barK_1 K_{21} +\barK_{21}K_1)+\alpha^2 \barK_{21} K_{21} &\prec& I.
\end{eqnarray}
This last equation is reminiscent to a Linear Matrix Inequality~\cite{el2000advances} (LMI, i.e., $\sum_i y_iS_i\succ 0$ with $y_i\in\bbR$ and $S_i\in\Sym(d,\bbR)$ where the coefficients $y_i$ are however {\em linked} between them).

Let us consider the special case of diagonal matrices corresponding to the polydisk domain: 
$K=\diag(k_1,\ldots,k_d)$ and $K'=\diag(k_1',\ldots,k_d')$ of the Siegel disk domain.

First, let us start with the simple case $d=1$, i.e.,  the Siegel disk $\SD(1)$ which is the complex open unit disk $\{ k\in\bbC\ :\ \bar{k}{k}<1\}$.
Let $k_\alpha=(1-\alpha)k_1+\alpha k_2=k_1+\alpha k_{21}$ with $k_{21}=k_2-k_1$. 
We have
$\bar{k}_\alpha k_\alpha=a\alpha^2+b\alpha+c$ with 
$a=\bar{k}_{21}k_{21}$, $b=\bar{k}_1k_{21}+\bar{k}_{21}k_1$ and $c=\bar{k}_1k_1$.
To find the two intersection points of line $(k_1k_2)$ with the boundary of $\SD(1)$, we need to solve $\bar{k}_\alpha k_\alpha=1$.
This amounts to solve an {\em ordinary quadratic equation} since all coefficients $a$, $b$, and $c$ are provably reals.
Let $\Delta=b^2-4ac$ be the discriminant ($\Delta>0$ when $k_1\not=k_2$).
We get the two solutions $\alpha_m=\frac{-b-\sqrt{\Delta}}{2a}$ and $\alpha_M=\frac{-b+\sqrt{\Delta}}{2a}$, and apply the 1D formula for the Hilbert distance: 
\begin{equation} 
\rho_K(k_1,k_2)= \frac{1}{2} \log \left(\frac{\alpha_M(1-\alpha_m)}{|\alpha_m|(\alpha_M-1)}\right).
\end{equation}
Doing so, we obtain a formula equivalent to Eq.~\ref{eq:DistKlein1D}.

For diagonal matrices with $d>1$, we get the following system of $d$ inequalities:
\begin{equation}
\alpha^2_i \left(\bar{k}_i'-\bar{k}_i\right)\left({k}_i'-{k}_i\right)+\alpha_i \left(\bar{k}_i(k_i'-k_i)+k_i(\bar{k}_i'-\bar{k}_i) \right)+
\bar{k}_ik_i-1  \leq 0, \forall i\in\{1,\ldots, d\}.
\end{equation}

For each inequality, we solve the quadratic equation as in the 1d case above, yielding two solutions $\alpha_i^-$ and $\alpha_i^+$.
Then we satisfy all those constraints by setting
\begin{eqnarray}
\alpha_- &=& \max_{i\in\{1,\ldots,d\}} \alpha_i^-,\\
\alpha_+ &=& \min_{i\in\{1,\ldots,d\}} \alpha_i^+,
\end{eqnarray}
and we compute the Hilbert distance:
\begin{equation} 
\rho_K(K_1,K_2)= \frac{1}{2} \log \left(\frac{\alpha_+(1-\alpha_-)}{|\alpha_-|(\alpha_+-1)}\right).
\end{equation}

\label{proofKlein1d}
\begin{Theorem}[Siegel-Klein distance for diagonal matrices]\label{thm:SKdiagformula}
The Siegel-Klein distance between two diagonal matrices in the Siegel-Klein disk can be calculated exactly in linear time.
\end{Theorem}

Notice that the proof extends to triangular matrices as well.

When the matrices are non-diagonal, we have to solve analytically the equation:

\begin{eqnarray}
&&\max\ |\alpha|,\\
\mbox{such that}&&\alpha^2 S_2+\alpha S_1+S_0 \prec 0,
\end{eqnarray}
with the following  {\em Hermitian} matrices (with all real eigenvalues):
\begin{eqnarray}
S_2 &=& \barK_{21} K_{21} =S_2^H,\\
S_1 &=& \barK_1 K_{21} +\barK_{21}K_1=S_1^H,\\
S_0 &=& \barK_1 K_1-I=S_0^H.
\end{eqnarray}

Although $S_0$ and $S_2$ commute, it is not necessarily the case for $S_0$ and $S_1$, or $S_1$ and $S_2$.

When $S_0$, $S_1$ and $S_2$ are  {\em simultaneously diagonalizable} via {\em congruence}~\cite{bustamante2019solving}, the optimization problem becomes:
\begin{eqnarray}
&&\max\ |\alpha|,\\
\mbox{such that}&&\alpha^2 D_2+\alpha D_1  \prec -D_0,
\end{eqnarray}
 where $D_i=P^\top S_iP$ for some $P\in\GL(d,\bbC)$, 
and we apply Theorem~\ref{thm:SKdiagformula}.
The same result applies for simultaneously diagonalizable matrices $S_0$, $S_1$ and $S_2$ via {\em similarity}: $D_i=P^{-1}S_i P$ with $P\in\GL(d,\bbC)$.


Notice that the Hilbert distance (or its squared distance) is {\em not} a separable distance, even in the case of diagonal matrices.
(But recall that the {\em squared} Siegel-Poincar\'e distance in the upper plane is separable for diagonal matrices.)

When $d=1$, we have
\begin{equation}
\rho_U(z_1,z_2)=\rho_D(w_1,w_2)=\rho_K(k_1,k_2).
\end{equation}

We now investigate a guaranteed fast  scheme for approximating the Siegel-Klein distance in the general case.

\subsection{A fast guaranteed approximation of the Siegel-Klein distance}\label{sec:SKdistapprox}
In the general case, we use the bisection approximation algorithm which is a {\em geometric approximation technique} that requires to only calculate  operator norms (and not the square root matrices required in the functions $\Phi_{\cdot}(\cdot)$ for calculating the Siegel distance in the disk domain).

We have the following key property of the Hilbert distance:

\begin{Property}[Bounding Hilbert distance]\label{prop:bh}
Let $\Omega_+\subset \Omega\subset \Omega_-$ be strictly nested open convex bounded domains.
Then we have the following inequality for the corresponding Hilbert distances:
\begin{equation}
H_{\Omega_+,\kappa}(p,q) \geq H_{\Omega,\kappa}(p,q) \geq H_{\Omega_-,\kappa}(p,q).
\end{equation}
\end{Property}

\begin{figure}
\centering
\includegraphics[width=0.75\columnwidth]{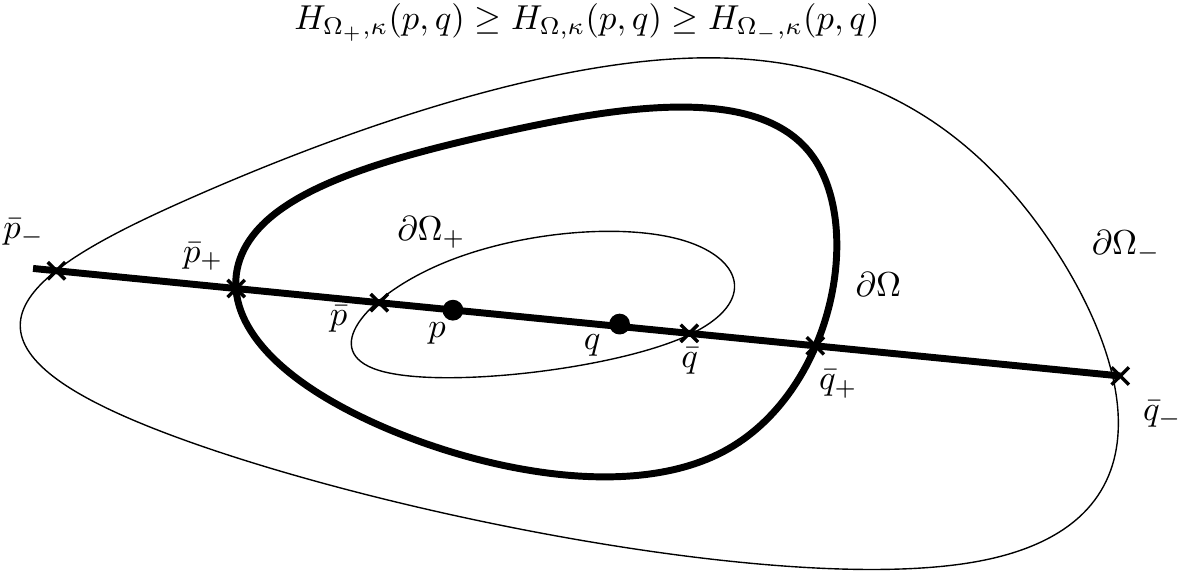}

\caption{Inequalities of the Hilbert distances induced by nested bounded open convex domains.\label{fig:nesteddomains}}
\end{figure}

Figure~\ref{fig:nesteddomains} illustrates the Property~\ref{prop:bh} of Hilbert distances corresponding to nested domains.
Notice that when $\Omega_-$ is a large enclosing ball of $\Omega$ with radius increasing to infinity, we have $\alpha_-\simeq \alpha_+$, and therefore the Hilbert distance tends to zero.

\begin{figure}
\centering
\includegraphics[width=0.65\columnwidth]{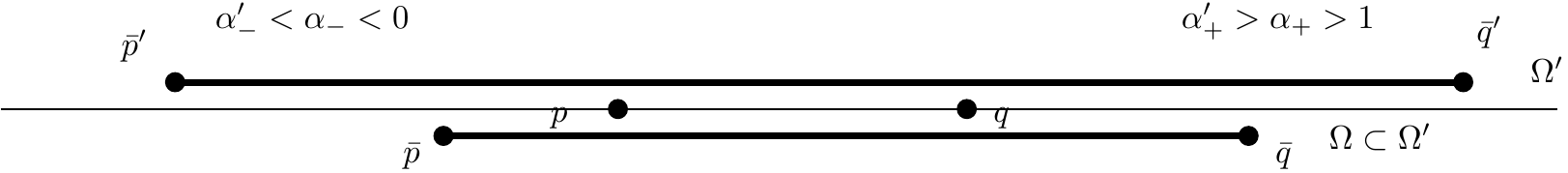}

\caption{Comparison of the Hilbert distances $H_{\Omega,\kappa}(p,q)$ and $H_{\Omega',\kappa}(p,q)$ induced by nested open interval domains $\Omega\subset\Omega'$: $H_{\Omega,\kappa}(p,q)\geq H_{\Omega',\kappa}(p,q)$.\label{fig:twonesteddomains}}
\end{figure}

\begin{proof}
Recall that $H_{\Omega,\kappa}(p,q)=H_{\Omega\cap (pq),\kappa}(p,q)$, i.e., the  Hilbert distance with respect to domain $\Omega$ can be calculated as an equivalent $1$-dimensional Hilbert distance by considering the open bounded (convex) {\em interval} $\Omega\cap (pq)=[\bar{p}\bar{q}]$.
Furthermore, we have $[\bar{p}\bar{q}]\subset [\bar{p}'\bar{q}']=\Omega'\cap (pq)$ (with set containment $\Omega\subset\Omega'$).
Therefore let us consider the 1D case as depicted in Figure~\ref{fig:twonesteddomains}.
Let us choose $p<q$ so that we have $\bar{p}'\leq \bar{p}<p<q<\bar{q}\leq \bar{q}'$.
In 1D, the Hilbert distance is expressed as
\begin{equation}
H_{\Omega,\kappa}(p,q) :=  \kappa \log \left(\frac{|\bar{q}-p|\ |\bar{p}-q|}{|\bar{q}-q|\ |\bar{p}-p|}\right),
\end{equation}
for a prescribed constant $\kappa>0$.
Therefore it follows that
\begin{equation}
H_{\Omega,\kappa}(p,q)-H_{\Omega',\kappa}(p,q) :=  \kappa \log \left(
\frac{|\bar{q}-p|\ |\bar{p}-q|}{|\bar{q}-q|\ |\bar{p}-p|} \times 
\frac{|\bar{q}'-q|\ |\bar{p}'-p|}{|\bar{q}'-p|\ |\bar{p}'-q|}
\right).
\end{equation}

We can rewrite the argument of the logarithm as follows:
\begin{eqnarray}
\frac{|\bar{q}-p|\ |\bar{p}-q|}{|\bar{q}-q|\ |\bar{p}-p|} \times 
\frac{|\bar{q}'-q|\ |\bar{p}'-p|}{|\bar{q}'-p|\ |\bar{p}'-q|} &=&
 \frac{(\bar{q}-p)(\bar{q}'-q)}{(\bar{q}-q)(\bar{q}'-p)} \times \frac{(p-\bar{p}')(q-\bar{p})}{(p-\bar{p})(q-\bar{p}')},\\
&=&\CR(\bar{q},\bar{q}';p,q) \times \CR(p,q;\bar{p}',\bar{p}),
\end{eqnarray}
with 
\begin{equation}
\CR(a,b;c,d) = \frac{|a-c|\ |b-d|}{|a-d|\ |b-c|} = \frac{ \frac{|a-c|}{|b-c|} }{ \frac{|a-d|}{|b-d|} }.
\end{equation}
Since $\bar{p}'\leq \bar{p}<p<q<\bar{q}\leq \bar{q}'$, we have $\CR(\bar{q},\bar{q}';p,q)\geq 1$ and $\CR(p,q;\bar{p}',\bar{p})\geq 1$, see~\cite{richter2011perspectives}.
Therefore we deduce that
$H_{\Omega,\kappa}(p,q)\geq H_{\Omega',\kappa}(p,q)$ when $\Omega\subset\Omega'$.
\end{proof}

Therefore the bisection search for finding the values of $\alpha_-$ and $\alpha_+$ yields both lower and upper bounds on 
the exact Siegel-Klein distance as follows:
Let $\alpha_-\in (l_-,u_-)$ and $\alpha_+\in (l_+,u_+)$ where $l_-$, $u_-$, $l_+$ , $u_+$ are real values defining the extremities of the intervals.
Using Property~\ref{prop:bh}, we get the following theorem:

\begin{Theorem}[Lower and upper bounds on the Siegel-Klein distance]\label{prop:LBboundSK}
The Siegel-Klein distance between two matrices $K_1$ and $K_2$ of the Siegel disk is bounded as follows:
\begin{equation}
\rho_K(l_-,u_+)  \leq \rho_K(K_1,K_2) \leq \rho_K(u_-,l_+), 
\end{equation}
where
\begin{equation} 
\rho_K(\alpha_m,\alpha_M) := \frac{1}{2} \log \left(\frac{\alpha_M(1-\alpha_m)}{|\alpha_m|(\alpha_M-1)}\right).
\end{equation}
\end{Theorem}

Figure~\ref{fig:LBUBHilbertKlein} depicts the guaranteed lower and upper bounds obtained by performing the bisection search for approximating the point $\bar{K}_1\in (\bar{K}_1'',\bar{K}_1')$ and the points $\bar{K}_2\in (\bar{K}_2',\bar{K}_2'')$.

We have:
\begin{equation}
\CR(\bar{K}_1',K_1;K_2,\bar{K}_2') \geq \CR(\bar{K}_1,K_1;K_2,\bar{K}_2) \geq \CR(\bar{K}_1'',K_1;K_2,\bar{K}_2''),
\end{equation}
where $\CR(a,b;c,d) = \frac{\|a-c\| \|b-d\|}{\|a-d\| \|b-c\|}$ denotes the cross-ratio. 
Hence we have
\begin{equation}
H_{\Omega',\frac{1}{2}}(K_1,K_2)  \geq \rho_K(K_1,K_2) \geq H_{\Omega'',\frac{1}{2}}(K_1,K_2).
\end{equation}

Notice that the approximation of the Siegel-Klein distance by line bisection requires only to calculate an operator norm $\|M\|_O$ at each step: This involves calculating the smallest and largest eigenvalues of $M$, or the largest eigenvalue of $M\bar{M}$.
To get a $(1+\epsilon)$-approximation, we need to perform $O(\log\frac{1}{\epsilon})$ dichotomic steps.
This yields a fast method to approximate the Siegel-Klein distance compared with the costly exact calculation of the Siegel-Klein distance of Eq.~\ref{eq:SKdist} which requires to calculate $\Phi_{\cdot}(\cdot)$ functions: This involves the calculation of a square root of a complex matrix.
Furthermore, notice that the operator norm can be numerically approximated using a Lanczos's power iteration scheme~\cite{LanczosApprox-1992,higham2010computing} (see also~\cite{li2016tight}).

\begin{figure}
\centering
\includegraphics[width=0.65\columnwidth]{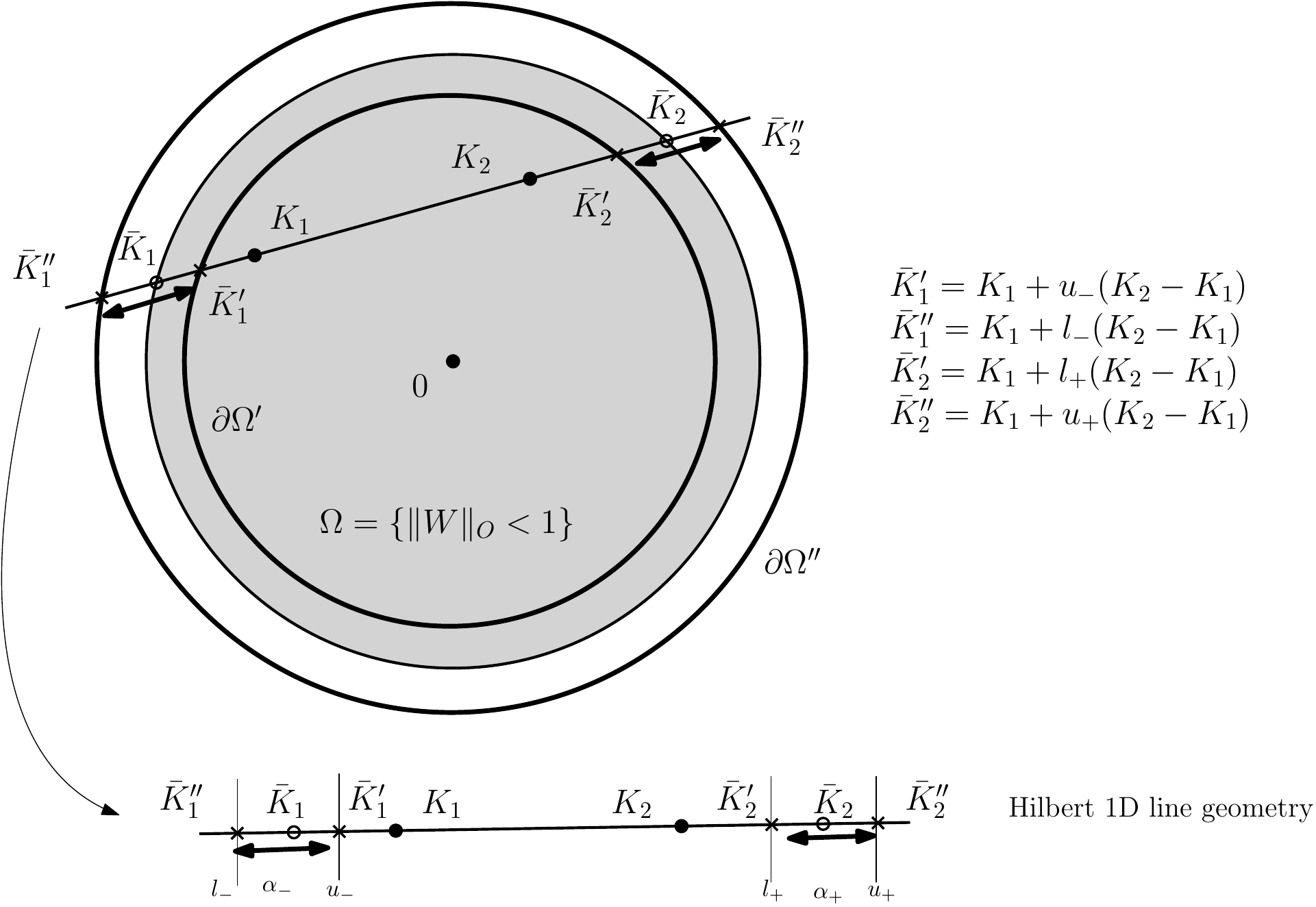}

\caption{Guaranteed lower and upper bounds for the Siegel-Klein distance by considering nested open matrix balls.\label{fig:LBUBHilbertKlein}}
\end{figure}

\subsection{Hilbert-Fr\"obenius distances and fast simple bounds on the Siegel-Klein distance}\label{sec:HF}

Let us notice that although the Hilbert distance does {\em not} depend on the chosen norm in the vector space,  the Siegel complex ball $\SD(d)$ 
{\em is} defined according to the operator norm.
In a finite-dimensional vector space, all norms are said ``equivalent'': 
That is, given two norms $\|\cdot\|_a$ and $\|\cdot\|_b$ of vector space $X$, there exists positive constants $c_1$ and $c_2$ such that
\begin{equation}
 c_1\|x\|_a \leq  \|x\|_a \leq  c_2\|x\|_b,\ \forall x\in X.
\end{equation} 

In particular, this property holds for the operator norm and Fr\"obenius norm of finite-dimensional complex matrices with positive constants $c_d$, $C_d$, $c_d'$ and $C_d'$ depending on the dimension $d$ of the square matrices:
\begin{eqnarray}
 c_d \|M\|_O  &\leq  \|M\|_F \leq&  C_d \|M\|_O,\quad \forall M\in M(d,\bbC),\\
 c_d' \|M\|_F  &\leq  \|M\|_O \leq&  C_d' \|M\|_F,\quad \forall M\in M(d,\bbC).
\end{eqnarray} 
As mentioned in the introduction, we have $\|M\|_O \leq \|M\|_F$.

Thus the Siegel ball domain $\SD(d)$ may be enclosed by an open Fr\"obenius ball $\FD\left(d,\frac{1}{(1+\epsilon)c_d}\right)$ (for any $\epsilon>0$)  with
 \begin{equation}
\FD(d,r) \eqdef \left\{
 M\in M(d,\bbC)\ :\ \|M\|_F < r
\right\}.
\end{equation}

Therefore we have 
\begin{equation}
 H_{\FD\left(d,\frac{1}{c_d}\right),\frac{1}{2}} (K_1,K_2) \leq \rho_K(K_1,K_2),
\end{equation}
where $H_{\FD(d,r),\frac{1}{2}}$ denotes the {\em Fr\"obenius-Klein distance}, i.e., the Hilbert distance induced by the Fr\"obenius balls $\FD(d,r)$ with constant $\kappa=\frac{1}{2}$.  

Now, we can calculate in closed-form the Fr\"obenius-Klein distance by computing the two intersection points of the line $(K_1K_2)$ 
with the Fr\"obenius ball $\FD(d,r)$. This amounts to solve an ordinary quadratic equation $\|K_1+\alpha(K_2-K_1)\|^2_F=r$ for parameter $\alpha$:
\begin{equation}
\|K_{21}\|_F^2 \alpha^2+ \left(\sum_{i,j} K_{21}^{i,j}\bar{K}_1^{i,j}+K_1^{i,j}\bar{K}_{21}^{i,j} \right)\alpha + (\|K_{1}\|_F^2-r) =0,
\end{equation}
where $K^{i,j}$ denotes the coefficient of matrix $K$ at row $i$ and column $j$.
Notice that $\left(\sum_{i,j} K_21^{i,j}\bar{K}_1^{i,j}+K_1^{i,j}\bar{K}_{21}^{i,j} \right)$ is a real.
Once $\alpha_-$ and $\alpha_+$ are found, we apply the 1D formula of the Hilbert distance of Eq.~\ref{eq:HilbertDist1D}.

We summarize the result as follows:

\begin{Theorem}[Lower bound on Siegel-Klein distance]\label{thm:lbskdist}
The Siegel-Klein distance is lower bounded by the Fr\"obenius-Klein distance for the unit complex Fr\"obenius ball, and it can be calculated in $O(d^2)$ time.
\end{Theorem}

\section{The smallest enclosing ball in the SPD manifold and in the Siegel spaces}\label{sec:SEB}

The goal of this section is to compare two implementations of a generalization of the Badoiu and Clarkson's algorithm~\cite{badoiu2003smaller} to approximate the Smallest Enclosing Ball (SEB) of a set of complex matrices:
The implementation using the Siegel-Poincar\'e disk (with respect to the Kobayashi distance $\rho_D$), and the implementation using the Siegel-Klein disk (with respect to the Siegel-Klein distance $\rho_K$).

In general, we may encode a pair of features $(S,P)\in\Sym(d,\bbR)\times\bbP_{++}(d,\bbR)$ in applications 
as a Riemann matrix $Z(S,P):=S+iP$, and consider the underlying geometry of the Siegel upper space. 
For example, anomaly detection of time-series maybe considered 
by considering $(\dot\Sigma(t),\Sigma(t))$ where $\Sigma(t)$ is the covariance matrix at time $t$ and $\dot\Sigma(t)\simeq \frac{1}{\dt}(\Sigma(t+\mathrm{d}t)-\Sigma(t))$ is the approximation of the derivative of the covariance matrix (a symmetric matrix) for a small prescrived 
value of $\mathrm{d}t$.

The generic  Badoiu and Clarkson's algorithm~\cite{badoiu2003smaller} (BC algorithm) for a set $\{p_1,\ldots, p_n\}$ of $n$   points in a metric space $(X,\rho)$ is described as follows:

\begin{itemize}
	\item Initialization: Let $c_1=p_1$ and $l=1$
	\item Repeat $L$ times:
	\begin{itemize}
	\item Calculate the farthest point: $f_l=\arg\min_{i\in[d]}\ \rho(c_{l},p_i)$.
	
	\item Geodesic cut: Let $c_{l+1}=c_{l}\#_{t_l} f_l$, where
	$p\#_{t_l} q$ is the point which satisfies 
	\begin{equation}
	\rho(p,p\#_{t_l}^X q) = t_l \rho(p,q).
	\end{equation}
	
	\item $l\leftarrow l+1$.
	\end{itemize}
	
\end{itemize}

This elementary SEB approximation algorithm has been instantiated to various metric spaces with proofs of convergence according to the sequence $\{t_l\}_l$:
see~\cite{hyperbolicSEB-2015} for the case of hyperbolic geometry, \cite{arnaudon2013approximating} for Riemannian geometry with bounded sectional curvatures, \cite{BCBregman-2005,nielsen2006approximating} for dually flat spaces (a non-metric space equipped with a Bregman divergences~\cite{nielsen2007visualizing,boissonnat2010bregman}), etc.
In Cartan-Hadamard manifolds~\cite{arnaudon2013approximating}, we require the series $\sum_i t_i$ to diverge while the series $\sum_i t_i^2$ to converge.
The number of iterations $L$ to get a $(1+\epsilon)$-approximation of the SEB depends on the underlying geometry and the sequence $\{t_l\}_l$.
For example, in Euclidean geometry, setting $t_l=\frac{1}{l+1}$ with $L=\frac{1}{\epsilon^2}$ steps  yield a $(1+\epsilon)$-approximation of the SEB~\cite{badoiu2003smaller}.

We start by recalling the Riemannian generalization of the BC algorithm, and then consider the Siegel spaces.

\subsection{Approximating the smallest enclosing ball in Riemannian spaces}

We first instantiate a particular example of Riemannian space, the space of Symmetric Positive-Definite matrix manifold (PD or SPD manifold for short), and then consider the general case on a Riemannian manifold $(M,g)$.

\subsubsection{Approximating the SEB on the SPD manifold}

Given $n$  positive-definite matrices~\cite{bougerol1993kalman,fletcher2011horoball} $P_1,\ldots, P_n$ of size $d\times d$,  we ask to calculate the SEB  with circumcenter $P^*$ minimizing the following objective function:
\begin{equation}
\min_{P\in \PD(d)}\max_{i\in \{1,\ldots,n\}}\ \rho_\PD(P,P_i).
\end{equation}
This is a minimax optimization problem.
The SPD cone is {\em not} a complete metric space with respect to the {\em Fr\"obenius distance}, but is a {\em complete metric space} with respect to the  {\em natural Riemannian distance}.

When the minimization is performed with respect to the Fr\"obenius distance, we can solve this problem using techniques of Euclidean computational geometry~\cite{boissonnat1998algorithmic,badoiu2003smaller} by {\em vectorizing} the PSD matrices $P_i$ into corresponding vectors $v_i=\vec(P_i)$ of $\bbR^{d\times d}$ such that
$\|P-P'\|_F=\|\vec(P)-\vec(P')\|_2$, where $\vec(\cdot): \Sym(d,\bbR)\rightarrow\bbR^{d\times d}$ vectorizes a matrix by stacking its column vectors. In fact, since the matrices are symmetric, it is enough to {\em half-vectorize} the matrices: $\|P-P'\|_F=\|\vec^+(P)-\vec^+(P')\|_2$, where $\vec^+(\cdot): \Sym_{++}(d,\bbR)\rightarrow\bbR^{\frac{d(d+1)}{2}}$ , see~\cite{nielsen2017fast}.

\begin{Property}
The smallest enclosing ball of a finite set of positive-definite matrices is unique.
\end{Property}

Let us mention the two following proofs:
\begin{itemize}
\item The SEB is well-defined and unique since the SPD manifold is a {\em Bruhat-Tits space}: That is, 
a complete metric space enjoying a semiparallelogram law:
For any $P_1,P_2\in\PD(d)$ and geodesic midpoint $P_{12}=P_1(P_1^{-1}P_2)^{\frac{1}{2}}$ (see below), we have:
\begin{equation}
\rho_\PD^2(P_1, P_2)+ 4 \rho_\PD^2(P, P_{12}) \leq 2 \rho_\PD^2(P, P_1)+2 d_\PD^2(P, P_2),\ \forall P\in\PD(d).
\end{equation}
See~\cite{lang2012math} page 83 or~\cite{bhatia2009positive} Chapter~6).
In a Bruhat-Tits space, the SEB is guaranteed to be unique~\cite{lang2012math,bruhat1972groupes}.

\item Another proof of the uniqueness of the SEB on a SPD manifold consists in noticing that the SPD manifold is a Cartan-Hadamard manifold~\cite{arnaudon2013approximating}, and the SEB on Cartan-Hadamard manifolds are guaranteed to be unique.
\end{itemize}

We shall use the  invariance property of the Riemannian distance by congruence: 
\begin{equation}
\rho_\PD\left(C^\top P_1C,C^\top P_2C\right)=\rho_\PD(P_1,P_2),\quad\forall C\in\GL(d,\bbR).
\end{equation}

In particular, choosing $C=P_1^{-\frac{1}{2}}$, we get
\begin{equation}
\rho_\PD(P_1,P_2)=\rho\left(I,P_1^{-\frac{1}{2}}P_2P_1^{-\frac{1}{2}}\right).
\end{equation}
The geodesic from $I$ to $P$ is $\gamma_{I,P}(\alpha)=\Exp(\alpha\Log P)=P^\alpha$.
The set  $\{\lambda_i(P^\alpha)\}$  of the $d$ eigenvalues of $P^\alpha$ coincide with the set 
 $\{\lambda_i(P)^\alpha\}$ of eigenvalues of $P$ raised to the power $\alpha$ (up to a permutation).

Thus to cut the geodesic $I\#_t^\PD P$, we have to solve the following problem:
\begin{equation}
\rho_\PD(I,P^\alpha)=t\times \rho_\PD(I,P).
\end{equation}
That is
\begin{eqnarray}
\sqrt{\sum_i \log^2 \lambda_{i}(P)^\alpha} &=& t\times \sqrt{ \sum_i \log^2 \lambda_{i}(P)},\\
\alpha \times \sqrt{\sum_i \log^2 \lambda_{i}(P)} &=& t\times \sqrt{ \sum_i \log^2 \lambda_{i}(P)}.
\end{eqnarray}
The solution is $\alpha=t$.
Thus $I\#_t^\PD P=P^t$.
For arbitrary $P_1$ and $P_2$, we first apply the congruence transformation with $C=P_1^{-\frac{1}{2}}$, 
use the solution $I\#_t^\PD CPC^\top =(CPC^\top)^t$, and apply the inverse congruence transformation with $C^{-1}=P_1^{\frac{1}{2}}$.
It follows the theorem:

\begin{Theorem}[Geodesic cut on the SPD manifold]\label{thm:SPDcut}
For any $t\in (0,1)$, we have the closed-form expression of the geodesic cut on the manifold of positive-definite matrices:
\begin{eqnarray}
P_1\#_t^\PD P_2 &=&   P_1^{\frac{1}{2}} \Exp\left(t\ \Log \left(P_1^{-\frac{1}{2}}P_2 P_1^{-\frac{1}{2}}\right) \right) P_1^{\frac{1}{2}},\\
&=&  P_1^{\frac{1}{2}} \left(P_1^{-\frac{1}{2}} P_2 P_1^{-\frac{1}{2}}\right)^t P_1^{\frac{1}{2}},\\
&=& P_1 (P_1^{-1}P_2)^t,\\
&=& P_2 (P_2^{-1}P_1)^{1-t}.
\end{eqnarray}
\end{Theorem}

The matrix $P_1^{-\frac{1}{2}} P_2 P_1^{-\frac{1}{2}}$ can be rewritten using the orthogonal eigendecomposition as $U D U^\top$, where 
$D$ is the diagonal matrix of generalized eigenvalues. Thus the PD geodesic can be rewritten as
\begin{equation}
P_1\#_t^\PD P_2 = P_1^{\frac{1}{2}} U  D^t U^\top P_1^{\frac{1}{2}}.
\end{equation}

We instantiate the generic algorithm to positive-definite matrices as follows:
\vskip 0.3cm

\begin{center}
\fbox{
\begin{minipage}{0.8\textwidth}
\noindent Algorithm $\mathrm{ApproximatePDSEB}(\{P_1,\ldots,P_n\},L)$:
\begin{itemize}
	\item Initialization: Let $C_1=P_1$ and $l=1$
	\item Repeat $L$ times:l
	\begin{itemize}
	\item Calculate the index of the farthest matrix: 
	\begin{equation*}
	f_l = \arg\min_{i\in\{1,\ldots, d\}}\ \rho_\PD(C_{t},P_i).
	\end{equation*}
	
	\item Geodesic walk: 
	
	\begin{equation*}
	C_{l+1}= C_l^{\frac{1}{2}} \left(C_l^{-\frac{1}{2}} P_{f_l} C_l^{-\frac{1}{2}}\right)^l C_l^{\frac{1}{2}}
	\end{equation*}
	
	\item $l\leftarrow l+1$.
	\end{itemize}
	
\end{itemize}
\end{minipage}
}
\end{center}

The complexity of the algorithm is in $O(d^3nT)$ where $T$ is the number of iterations, $d$ the row dimension of the square matrices $P_i$ and $n$ the number of matrices.

Observe that the solution corresponds to  the  arc-length parameterization of the geodesic with boundary values on the SPD manifold:
\begin{equation}
\gamma_{P_1,P_2}(t)=P_1^{\frac{1}{2}} \exp(t \Log(P_1^{-\frac{1}{2}} P_2 P_1^{-\frac{1}{2}})) P_1^{\frac{1}{2}}.
\end{equation}

The curve $\gamma_{P_1,P_2}(t)$ is a geodesic for any affine-invariant metric distance $\rho_\psi(P_1,P_2)=\|\Log P_1^{-\frac{1}{2}} P_2 P_1^{-\frac{1}{2}}\|_\psi$ where $\|M\|_\psi=\psi(\lambda_1(M),\ldots,\lambda_d(M))$ is a symmetric gauge norm~\cite{mostajeran2019affine}.

In fact, we have shown the following property:

\begin{Property}[Riemannian geodesic cut]\label{prop:riegeocut}
Let $\gamma_{p,q}(t)$ denote the Riemannian geodesic linking $p$ and $q$ on a Riemannian manifold $(\calM,g)$ (i.e., parameterized proportionally to the arc-length and with respect to the Levi-Civita connection induced by the metric tensor $g$). Then we have
\begin{equation}
p_1\#_t^g p_2=\gamma_{p_1,p_2}(t)=\gamma_{p_2,p_1}(1-t).
\end{equation}
\end{Property}

Thus it follows the following generic Riemannian algorithm:
\vskip 0.3cm

\begin{center}
\fbox{
\begin{minipage}{0.8\textwidth}
\noindent Algorithm $\mathrm{ApproximateRieSEB}(\{p_1,\ldots,p_n\},g,L)$:
\begin{itemize}
	\item Initialization: Let $c_1=p_1$ and $l=1$
	\item Repeat $L$ times:
	\begin{itemize}
	\item Calculate the index of the farthest point: 
	\begin{equation*}
	f_l = \arg\min_{i\in\{1,\ldots, d\}}\ \rho_g(c_{l},p_i).
	\end{equation*}
	
	\item Geodesic walk: 
	
	\begin{equation*}
	c_{l+1}=  \gamma_{c_l,p_{f_l}}\left(t_l\right).
	\end{equation*}
	
	\item $l\leftarrow l+1$.
	\end{itemize}
	
\end{itemize}
\end{minipage}
}
\end{center}

Theorem~1 of~\cite{arnaudon2013approximating} guarantees the convergence of the {\sc ApproximateRieSEB} algorithm provided that we have a lower bound and an upper bound on the sectional curvatures of the manifold $(M,g)$.
The sectional curvatures of the PD manifold have been proven to be negative~\cite{helgason1979differential}.
The SPD manifold is a Cartan-Hadamard manifold with  scalar curvature  $\frac{1}{8}d(d+1)(d+2)$~\cite{andai2004information} depending on the dimension $d$ of the matrices.
Notice that we can identify $P\in\PD(d)$ with an element of the quotient space $\GL(d,\bbR)/O(d)$ since $O(d)$ is the isotropy subgroup of the $\GL(d,\bbR)$ for the action $P\mapsto C^\top PC$ (i.e., $I\mapsto C^\top IC=I$ when $C\in O(d)$).
Thus we have $\PD(d)\cong \GL(d,\bbR)/O(d)$.
The SEB with respect to the Thompson metric
\begin{equation}
\rho_T(P_1,P_2) := \max\left\{\log\lambda_\max(P_2P_1^{-1}), \log\lambda_\max(P_1P_2^{-1})\right\}
\end{equation}
 has been studied in~\cite{mostajeran2019affine}.

\subsection{Implementation in the Siegel-Poincar\'e disk}\label{sec:seb:SiegelPoincare}

Given $n$ $d\times d$ complex matrices $W_1,\ldots, W_n\in\SD(d)$,  we ask to find the smallest-radius enclosing ball with center $W*$ minimizing the following objective function:
\begin{equation}
\min_{W\in \SD(d)}\max_{i\in \{1,\ldots,n\}}\ \rho_D(W,W_i).
\end{equation}

This problem may have potential applications in image morphology~\cite{StructureTensorImageFiltering-2014} or anomaly detection of covariance matrices~\cite{tavallaee2008novel}.
We may model the dynamics of a covariance matrix time-series $\Sigma(t)$ by the representation $(\Sigma(t),\dot\Sigma(t))$ where
 $\dot\Sigma(t)= \frac{d}{\mathrm{d}{t}}\Sigma(t)\in\Sym(d,\bbR)$ and use the Siegel SEB to detect anomalies, see~\cite{cont2010information} for detection anomaly based on Bregman SEBs.

The Siegel-Poincar\'e upper plane and disk are {\em not} Bruhat-Tits space, but spaces of  non-positive curvatures~\cite{d1983characterization}. Indeed, when $d=1$, the Poincar\'e disk is not a Bruhat-Space.

Notice that when $d=1$, the hyperbolic ball in the Poincar\'e disk have Euclidean shape.
This is not true anymore when $d>1$:
Indeed, the equation of the ball centered at the origin $0$:
\begin{equation}
\mathrm{Ball}(0,r)= \left\{
 W\in\SD(d)\ :\ \log\left( \frac{1+\|W\|_O}{1-\|W\|_O}\right) \leq r
\right\},
\end{equation}
amounts to
\begin{equation}
\mathrm{Ball}(0,r)= \left\{ W\in\SD(d)\ :\ \|W\|_O  \leq \frac{e^r-1}{e^r+1} \right\}.
\end{equation}
When $d=1$, $\|W\|_O=|w|=\|(\Re(w),\Im(w))\|_2$, and Poincar\'e balls have Euclidean shapes. 
Otherwise, when $d>1$, $\|W\|_O=\sigma_{\max}(W)$ and $\sigma_{\max}(W) \leq  \frac{e^r-1}{e^r+1}$ is not a complex Fr\"obenius ball.

In order to apply the generic algorithm, we need to implement the geodesic cut operation $W_1\#_t W_2$.
We consider the complex symplectic map $\Phi_{W_1}(W)$ in the Siegel disk that maps $W_1$ to $0$ and $W_2$ to $W_2'=\Phi_{W_1}(W_2)$.
Then the geodesic between $0$ and $W_2'$ is   a straight line.

We need to find $\alpha(t) W=0\#_t^\SD W$ (with $\alpha(t)>0$) such that $\rho_D(0,\alpha(t) W)=t \rho_D(0,W)$.
That is, we shall solve the following equation:
\begin{equation}
\log \left(\frac{1+\alpha(t) \|W\|_O}{1-\alpha(t) \|W\|_O}\right) = t\times \log\left(\frac{1+  \|W\|_O}{1-\|W\|_O}\right).
\end{equation}
We find the exact solution as
\begin{equation}\label{eq:GeodesicCutSiegel0}
\alpha(t) = \frac{1}{\|W\|_O}  \frac{ (1+  \|W\|_O)^t - (1-  \|W\|_O)^t}{ (1+  \|W\|_O)^t + (1- \|W\|_O)^t}.
\end{equation}

\begin{Proposition}[Siegel-Poincar\'e geodesics from the origin]\label{prop:SPgeodesicOrigin}
The geodesic in the Siegel disk is 
\begin{equation}
\gamma_{0,W}^\SD(t)=\alpha(t)W
\end{equation}
with
$$
\alpha(t) = \frac{1}{\|W\|_O}  \frac{ (1+  \|W\|_O)^t - (1-  \|W\|_O)^t}{ (1+  \|W\|_O)^t + (1-  \|W\|_O)^t}.
$$
\end{Proposition}

Thus the midpoint $W_1\#^\SD W_2:=W_1\#^\SD_{\frac{1}{2}} W_2$ of $W_1$ and $W_2$ can be found as follows:
\begin{equation}
W_1\#^\SD W_2 = \Phi_{W_1}^{-1}\left(0\#^\SD \Phi_{W_1}(W_2)\right),
\end{equation}
where
\begin{eqnarray}
0\#^\SD W &=& \alpha\left(\frac{1}{2}\right)W,\\
&=&   \frac{1}{\|W\|_O}  \frac{ \sqrt{1+  \|W\|_O} - \sqrt{1-  \|W\|_O}}{ \sqrt{1+  \|W\|_O} + \sqrt{1-  \|W\|_O}} W.
\end{eqnarray}

To summarize, the algorithm recenters at every step the current center $C_t$ to the Siegel disk origin $0$:

\begin{center}
\fbox{
\begin{minipage}{0.8\textwidth}
\noindent Algorithm $\mathrm{ApproximateSiegelSEB}(\{W_1,\ldots,W_n\})$:
\begin{itemize}
	\item Initialization: Let $C_1=0$ and $l=1$.
	
	\item Compute $W_i'=\Phi_{C_1}(W_i)$ for all $i\in \{1,\ldots, n\}$.
	
	\item Repeat $L$ times:
	\begin{itemize}
	\item Calculate the index of the farthest point: $F_l=\arg\min_{i\in[d]} \rho_D(0,W_i')$.
	
	\item Geodesic cut: Let $C_{l+1}=0\#_{t_l}^\SD W_{F_l}$.
	
	\item Recenter $C_{l+1}$ to the origin for the next iteration: 
	Compute $W_i'=\Phi_{C_{l+1}}(W_i)$ for all $i\in \{1,\ldots,n\}$.
	Set $C_{l+1}=0$.
	
	\item $l\leftarrow l+1$.
	\end{itemize}
	
	\item Let the approximate circumcenter be mapped back to be consistent with the input:
	\begin{equation}
	\tilde{C}=\Phi_{C_{1}}^{-1}(\Phi_{C_{2}}^{-1}(\ldots \Phi_{C_L}^{-1}(0))\ldots).
	\end{equation}
	This amounts to calculate the symplectic map associated to the matrix
	$S=C_1^{(-1)}\times \ldots\times C_L^{(-1)}$. 
	Overall it costs $L$ matrix multiplications plus the cost of evaluation of the symplectic map defined by $S$.
	
\end{itemize}
\end{minipage}
}
\end{center}

The farthest point to the current approximation of the circumcenter can be calculated using the data-structure of the Vantage Point Tree (VPT), see~\cite{nielsen2009bregmanVPT}.

The Riemannian curvature tensor of the Siegel space is non-positive~\cite{Siegel-1943,hua1954estimation} and the sectional curvatures are non-positive~\cite{d1983characterization} and bounded above by a negative constant.
In our implementation, we chose the step sizes $t_l=\frac{1}{l+1}$.
Barbaresco~\cite{Barbaresco-MIG-2013} also adopted this iterative recentering operation for calculating the median in the Siegel disk.
However at the end of his algorithm, he does not map back the median  among the source matrix set.
Recentering is costly because we need to calculate a square root matrix to calculate $\Phi_{C}(W)$.
A great advantage of Siegel-Klein space is that we have straight geodesics anywhere in the disk so we do not need to perform recentering.

\subsection{Fast implementation in the Siegel-Klein disk}\label{sec:SEB-SK}

The main advantage of implementing the Badoiu and Clarkson's algorithm~\cite{badoiu2003smaller} in the Siegel-Klein disk is to avoid to perform the costly recentering operations (which require calculation of square root matrices).
Moreover, we do not have to roll back our approximate circumcenter at the end of the algorithm.

First, we state the following expression of the geodesics in the Siegel disk:

\begin{Proposition}[Siegel-Klein geodesics from the origin]\label{prop:SKgeodesicOrigin}
The geodesic from the origin in the Siegel-Klein disk is expressed
\begin{equation}
\gamma_{0,K}^\SK(t)=\alpha(t)K
\end{equation}
with
\begin{equation}
\alpha(t) = \frac{1}{\|K\|_O}  \frac{ (1+ \|K\|_O)^t - (1- \|K\|_O)^t}{ (1+ \|K\|_O)^t + (1- \|K\|_O)^t}.
\end{equation}
\end{Proposition}

The proof follows straightforwardly from Proposition~\ref{prop:SPgeodesicOrigin} because we have
$\rho_K(0,K)=\frac{1}{2}\rho_D(0,K)$.

\section{Conclusion and perspectives}\label{sec:concl}

In this work, we have generalized the Klein model of hyperbolic geometry to the Siegel disk domain of complex matrices by considering the Hilbert geometry induced by the Siegel disk, an open bounded convex complex matrix domain.
We compared this Siegel-Klein disk model with its Hilbert distance called the {\em Siegel-Klein distance} $\rho_K$ to both the Siegel-Poincar\'e disk model (Kobayashi distance $\rho_W$) and the Siegel-Poincar\'e upper plane (Siegel distance $\rho_U$).
We show how to convert  matrices $W$ of the Siegel-Poincar\'e disk model into equivalent matrices $K$ of Siegel-Klein disk model and matrices $Z$ in the Siegel-Poincar\'e upper plane via symplectic maps.
When the dimension $d=1$, we have the following equivalent hyperbolic distances: 
\begin{equation}
\rho_D(w_1,w_2)=\rho_K(k_1,k_2)=\rho_U(z_1,z_2).
\end{equation}

Since the geodesics in the Siegel-Klein disk are by construction straight, this model is well-suited to implement techniques of computational geometry~\cite{boissonnat1998algorithmic}.
Furthermore, the calculation of  the Siegel-Klein disk does {\em not} require to recenter one of its arguments to the disk origin, a computationally costly Siegel translation operation.
We reported a linear-time algorithm for computing the exact Siegel-Klein distance $\rho_K$ between diagonal matrices of the disk (Theorem~\ref{thm:SKdiagformula}), and a
 fast way to numerically approximate the Siegel distance by bisection searches with {\em guaranteed} lower and upper bounds (Theorem~\ref{prop:LBboundSK}).
Finally, we  demonstrated the algorithmic advantage of using the Siegel-Klein disk model instead of the Siegel-Poincar\'e disk model for approximating the smallest-radius enclosing ball of a finite set of complex matrices in the Siegel disk.
 In future work, we shall consider more generally the Hilbert geometry of homogeneous complex domains  and investigate quantitatively the Siegel-Klein geometry in applications ranging from  radar processing~\cite{Barbaresco-MIG-2013},
 image morphology~\cite{SiegelDescriptor-2016},  computer vision, to machine learning~\cite{RiemannTheta-2020}.
For example, the fast and robust guaranteed approximation of the Siegel-Klein distance may proved useful 
 for performing clustering analysis in image morphology~\cite{StructureTensorImageFiltering-2014,angulo2014morphological,SiegelDescriptor-2016}. 

\vskip 0.5cm
\noindent Additional material is available online at 
\begin{verbatim}
https://franknielsen.github.io/SiegelKlein/
\end{verbatim}

\vskip 1cm
\noindent {\bf Acknowledgments}: The author would like to thank Marc Arnaudon, Fr\'ed\'eric Barbaresco, Yann Cabanes, and Ga\"etan Hadjeres for fruitful discussions, pointing out several relevant references, and feedback related to the Siegel domains.

\appendix

\section*{Notations and main formulas}

{\small

\begin{supertabular}{ll}
\underline{Complex matrices}:\\ 
Number field $\bbF$ & Real $\bbR$ or complex $\bbC$\\
$M(d,\bbF)$ & Space of square $d\times d$ matrices in $\bbF$\\
$\Sym(d,\bbR)$  & Space of real symmetric matrices\\
$0$ & matrix with all coefficients equal to zero (disk origin)\\
Fr\"obenius norm & $\|M\|_F =\sqrt{\sum_{i,j} |M_{i,j}|^2}$\\
Operator norm & $\|M\|_O=\sigma_\max(M)=\max_i \{|\lambda_i(M)|\}$\\
\underline{Domains}: &\\
Cone of SPD matrices &  $\PD(d,\bbR)=\{ P\succ 0 \ :\ P\in\Sym(d,\bbR)\}$\\
Siegel-Poincar\'e upper plane & $\SH(d) = \left\{ Z = X + iY \st X\in\Sym(d,\bbR), Y\in\PD(d,\bbR)\right\}$\\
Siegel-Poincar\'e disk & $\SD(d) = \left\{  W\in\Sym(d,\bbC) \st I-\barW W\succ  0\right\}$\\
\underline{Distances}:&\\
Siegel distance & $\rho_U(Z_1,Z_2) = \sqrt{\sum_{i=1}^d \log^2\left(\frac{1+\sqrt{r_i}}{1-\sqrt{r_i}}\right)}$\\
& $r_i=\lambda_i\left(R(Z_1,Z_2)\right)$\\
& $R(Z_1,Z_2)  := (Z_1-Z_2)(Z_1-\barZ_2)^{-1} (\barZ_1 -\barZ_2) (\barZ_1 -\barZ_2)^{-1}$\\
Upper plane metric & $\ds_U (Z) = 2 \tr\left(Y^{-1}\dZ\ Y^{-1}\dZbar\right)$\\
PD distance & $\rho_\PD(P_1,P_2) = \|\Log(P_1^{-1}P_2)\|_F= \sqrt{\sum_{i=1}^d \log^2\left(\lambda_i(P_1^{-1}P_2)\right)}$\\
PD metric & $\ds_\PD(P)=\tr\left((P^{-1}\dP)^2\right)$\\
Kobayashi distance  & $\rho_D(W_1,W_2) = \log\left( \frac{1+\|\Phi_{W_1}(W_2)\|_O}{1-\|\Phi_{W_1}(W_2)\|_O} \right)$\\
Translation in the disk	& $\Phi_{W_1}(W_2)=(I-W_1\barW_1)^{-\frac{1}{2}} (W_2-W_1) (I-\barW_1W_2)^{-1} (I-\barW_1W_1)^{\frac{1}{2}}$\\
Disk distance to origin & $\rho_D(0,W) =  \log \left(\frac{1+\|W\|_O}{1-\|W\|_O}\right)$\\
Siegel-Klein distance & $\rho_K(K_1,K_2)=\left\{
\begin{array}{ll}
\frac{1}{2} \log \left|\frac{\alpha_+(1-\alpha_-)}{\alpha_-(\alpha_+-1)}\right|, & K_1\not=K_2,\\
0 & K_1=K_2
\end{array}
\right.$\\
& $\|(1-\alpha_-)K_1+\alpha_- K_2\|_O=1$ ($\alpha_-<0$),  $\|(1-\alpha_+)K_1+\alpha_+ K_2\|_O=1$ ($\alpha_+>1$)  \\
Seigel-Klein distance to $0$ &  $\rho_K(0,K)= \frac{1}{2} \log\left( \frac{1+\|K\|_O}{1-\|K\|_O}\right)$\\
\underline{Symplectic maps and groups}:\\
Symplectic map & $\phi_S(Z)=(AZ+B)(CZ+D)^{-1}$ with $S\in\Sp(d,\bbR)$ (upper plane)\\
& $\phi_S(W)$ with $S\in\Sp(d,\bbC)$ (disk)\\
Symplectic group & $\Sp(d,\bbF)=\left\{
 \mattwotwo{A}{B}{C}{D},   AB^\top=BA^\top, CD^\top=DC^\top,\quad AD^\top-BC^\top=I
\right\}$\\
 & $A,B,C,D\in\ M(d,\bbF)$\\
group composition law & matrix multiplication\\
group inverse law & $S^{(-1)}\eqdef  \mattwotwo{D^\top}{-B^\top}{-C^\top}{A^\top}$\\
Translation in $\bbH(d)$ of $Z=A+iB$ to $iI$ & $T_U(Z)=\mattwotwo{(B^{\frac{1}{2}})^\top}{0}{-(AB^{-\frac{1}{2}})^\top}{(B^{-\frac{1}{2}})^\top}$\\
 symplectic orthogonal matrices &  $\SpO(2d,\bbR) = \left\{ \mattwotwo{A}{B}{-B}{A} \ :\ A^\top A+B^\top B=I, A^\top B\in\Sym(d,\bbR) \right\}$\\
(rotations in $\bbSH(d)$) &\\
Translation to $0$ in $\SD(d)$ & $\Phi_{W_1}(W_2)=(I-W_1\barW_1)^{-\frac{1}{2}} (W_2-W_1) (I-\barW_1W_2)^{-1} (I-\barW_1W_1)^{\frac{1}{2}}$\\
$\Isom^+(\bbS)$ & Isometric orientation preserving group of  generic space $\bbS$\\
$\Moeb(d)$ & group of M\"obius transformations\\
\end{supertabular}

}

\section{The deflation method: Approximating the eigenvalues}\label{sec:deflation}

A matrix $M\in M(d,\bbC)$ is diagonalizable  if there exists a non-singular matrix $P$ and a diagonal matrix $\Lambda=\diag(\lambda_1,\ldots,\lambda_d)$ such that $M=P\Lambda P^{-1}$.
A Hermitian matrix (i.e., $M=M^*:=\bar{M}^\top$) is a diagonalizable self-adjoint matrix which has all real eigenvalues and admits a basis of orthogonal eigenvectors.
We can compute all eigenvalues of a Hermitian matrix $M$  by repeatedly
applying the (normalized) power method using the so-called deflation method (see~\cite{loehr2014advanced}, Chapter 10, and~\cite{izaac2018computational}, Chapter 7).
The {\em deflation method} proceeds iteratively to calculate numerically the (normalized) eigenvalues $\lambda_i$'s and eigenvectors $v_i$'s as follows:

\begin{enumerate}
\item Let $l=1$ and $M_1=M$.
\item Initialize at random a normalized vector $x_0\in\bbC^d$ (i.e., $x_0$ on the unit sphere with $x_0^*x_0=1$) 
\item For $j$ in  $(0,\ldots, L_l-1)$: 
\begin{eqnarray*}
x_{j+1}&\leftarrow& M_lx_j \\
x_{j+1}&\leftarrow& \frac{x_{j+1}}{x_{j+1}^*x_{j+1}}
\end{eqnarray*}
\item Let $v_l=x_{L_l}$ and $\lambda_l=x_{L_l}^* M_l x_{L_l}$
\item Let $l\leftarrow l+1$. If $l\leq d$ then let $M_l=M_{l-1}-\lambda_{l-1}v_{l-1}v_{l-1}^*$ and goto 2.
\end{enumerate}

The deflation method reports the eigenvalues $\lambda_i$'s such that
$$
|\lambda_1|>|\lambda_2|\geq \ldots\geq |\lambda_d|,
$$
where $\lambda_1$ is the {\em dominant eigenvalue}.

The overall number of normalized power iterations is $L=\sum_{i=1}^d L_i$ (matrix-vector multiplication), 
where the number of iterations of the {\em normalized power method} at stage $l$ can be defined such that we have
 $|x_{L_l}-x_{L_l-1}|\leq \epsilon$, for a prescribed value of $\epsilon>0$. 
Notice that the numerical errors of the eigenpairs $(\lambda_i,v_i)$'s propagate and accumulate at each stage.
That is, at stage $l$, the deflation method calculates the dominant eigenvector on a residual perturbated matrix $M_l$. 
The overall approximation of the eigendecomposition can be appreciated by calculating the last residual matrix:
\begin{equation}
\left\| M-\sum_{i=1}^d \lambda_i v_{i}v_{i}^*\right\|_F.
\end{equation}

The normalized power method exhibits {\em linear convergence} for diagonalizable matrices and {\em quadratic convergence} for Hermitian matrices.
Other numerical methods for numerically calculating the eigenvalues include the {\em Krylov subspace techniques}~\cite{loehr2014advanced,sun2016finite,izaac2018computational}.

\section*{Snippet code}\label{sec:snippet}

We implemented our software library and smallest enclosing ball algorithms in Java\texttrademark{}.

The code below is a snippet written in {\sc Maxima}: A computer algebra system, freely downloadable at \url{http://maxima.sourceforge.net/}

{\small
\begin{verbatim}
/* Code in Maxima */
/* Calculate the Siegel metric distance in the Siegel upper space */

load(eigen);

/* symmetric */
S1: matrix( [0.265,   0.5],
    [0.5 , -0.085]);

/* positive-definite */
P1: matrix( [0.235,   0.048],
    [0.048 ,  0.792]);

/* Matrix in the Siegel upper space */
Z1: S1+%i*P1;

S2:  matrix( [-0.329,  -0.2],
   [-0.2 , -0.382]);

P2: matrix([0.464,   0.289],
    [0.289  , 0.431]);

Z2: S2+%i*P2;

/* Generalized Moebius transformation */
R(Z1,Z2) := 
((Z1-Z2).invert(Z1-conjugate(Z2))).((conjugate(Z1)-conjugate(Z2)).invert(conjugate(Z1)-Z2));

R12: ratsimp(R(Z1,Z2));
ratsimp(R12[2][1]-conjugate(R12[1][2]));

/* Retrieve the eigenvalues: They are all reals */
r: float(eivals(R12))[1];

/* Calculate the Siegel distance */
distSiegel: sum(log( (1+sqrt(r[i]))/(1-sqrt(r[i]))  )**2, i, 1, 2);
\end{verbatim}
}

\bibliographystyle{plain}

\bibliography{SiegelHilbertGeometryBIBV4}

\end{document}